\documentclass[a4paper,english,cleveref,autoref,final]{lipics-v2019}

\newif\ifproc
\procfalse

\nolinenumbers

\usepackage{xspace}
\usepackage{bm}
\usepackage{hyperref}
\usepackage{csquotes}
\usepackage{color}
\usepackage{microtype}
\usepackage{nicefrac}
\usepackage{mathtools}
\usepackage{etoolbox}
\BeforeBeginEnvironment{enumerate*}{%
	\setlist[enumerate,1]{label=(\alph*), font={\bfseries}}
}
\AfterEndEnvironment{enumerate*}{%
	\setlist[enumerate,1]{label=\arabic*.}
}

\usepackage[sort&compress,numbers]{natbib}

\usepackage{tikz}
\usetikzlibrary{fit,patterns,decorations.pathreplacing}
\tikzset{every picture/.style={line width=0.75pt}} 

\title{Parameterized Complexity of Fair Many-to-One Matchings} 

\titlerunning{}

\author{Ramin Javadi}{Isfahan University of Technology, Isfahan, Iran}{rjavadi@iut.ac.ir}{https://orcid.org/0000-0003-4401-2110}{}

\author{Hossein Shokouhi}{Isfahan University of Technology, Isfahan, Iran}{h.shokouhizadeh@math.iut.ac.ir}{}{}

\authorrunning{R.\ Javadi and H. Shokouhi}

\Copyright{Ramin Javadi and Hossein Shokouhi}



\category{}

\supplement{}

\acknowledgements{
}

\ifproc
\relatedversion{A full version of the paper is available at \url{https://arxiv.org/abs/...}.}
\fi

\ifproc\else
\hideLIPIcs  
\fi

\newcommand{\prob}[3]{
	\begin{center}
		\fbox{~\begin{minipage}{.97\textwidth}
			\vspace{2pt}
			\noindent
			\normalsize\textsc{#1}

			\vspace{4pt}
			\setlength{\tabcolsep}{3pt}
			\renewcommand{\arraystretch}{1.0}
			\begin{tabularx}{\textwidth}{@{}lX@{}}
				\normalsize\textbf{Input:}	& \normalsize#2 \\
				\normalsize\textbf{Question:}		 & \normalsize#3
			\end{tabularx}
		\end{minipage}}
	\end{center}
}

\DeclareMathOperator{\col}{col}
\DeclareMathOperator{\vc}{vc}
\DeclareMathOperator{\fvs}{fvn}
\DeclareMathOperator{\fes}{fen}
\DeclareMathOperator{\tw}{tw}
\DeclareMathOperator{\cw}{cw}
\DeclareMathOperator{\mw}{mw}
\DeclareMathOperator{\pw}{pw}
\DeclareMathOperator{\nd}{nd}
\DeclareMathOperator{\td}{td}
\DeclareMathOperator{\vi}{vi}
\newcommand{\gfm}{\textsc{Generalized Fair Matching}}
\newcommand{\fm}{\textsc{Fair Matching}}
\newcommand{\ubp}{\textsc{Unary Bin Packing}}
\newcommand{\mcc}{\textsc{Multicolored Clique}}

\begin{document}

\maketitle

\begin{abstract}
Given a bipartite graph $G=(U\cup V,E)$, a left-perfect many-to-one matching is a subset $M \subseteq E$ such that each vertex in $U$ is incident with exactly one edge in $M$. If $U$ is partitioned into some groups, the matching is called fair if for every $v\in V$, the difference between the number of vertices matched with $v$ in any two groups does not exceed a given threshold. In this paper, we investigate parameterized complexity of fair left-perfect many-to-one matching problem with respect to the structural parameters of the input graph. In particular, we prove that the problem is W[1]-hard with respect to the feedback vertex number, tree-depth and the maximum degree of $U$, combined. Also, it is W[1]-hard with respect to the path-width, the number of groups and the maximum degree of $U$, combined. In the positive side, we prove that the problem is FPT with respect to the treewidth and the maximum degree of $V$. Also, it is FPT with respect to the neighborhood diversity of the input graph (which implies being FPT with respect to vertex cover and modular-width). Finally, we prove that the problem is FPT with respect to the tree-depth and the number of groups. 
\end{abstract}



\section{Introduction}
A left-perfect many-to-one matching for a bipartite graph $G$ with a bipartition $(U,V)$ is a subset of edges $M\subseteq E(G)$ such that each vertex in $U$ is incident with exactly one edge in $M$. In this paper, we study the problem of finding a left-perfect many-to-one matching which satisfies a fairness Max-Min criterion. In particular, the set $U$ is colored with a set of colors $C$ and we aim to find a  left-perfect many-to-one matching $M$ such that for each vertex $v\in V$, the number of  vertices matched to $v$ of all colors are  almost equitable with a prescribed tolerance, i.e. the difference of most and least frequent colors of vertices matched with each vertex in $V$ is bounded by a given constant.  
The concept of fairness in matching problems have many applications in different assignment problems in which you are going to allocate to each object (e.g. jobs, schools, hospitals, constituencies) a number of representatives (e.g. employees, students, doctors, voters) such that the assigned items to each object are fairly distributed with some sense such as age, ethnic, region, type, etc. Computational complexity of this problem has been firstly studied in \cite{main} and it is proved that the problem is NP-hard even when $|C|=\Delta_U=\Delta_V=3$, where $|C|$ is the number of colors, $\Delta_U$ and $\Delta_V$ are the maximum degree of vertices in $U$ and $V$, respectively. It is also proved in \cite{main} that the problem is FPT with respect to $|V|$ (see Theorem~\ref{thm:FPTk}) and polynomial solvable for complete bipartite graphs. We will generalize two former results by proving that the problem is FPT with respect to the neighborhood diversity of the input graph. In \cite{main}, fair matching problem is also studied with another fairness measure called \textit{margin of victory} (MoV) which is defined as the difference of most and second most frequent colors in matched vertices of each vertex in $V$. 

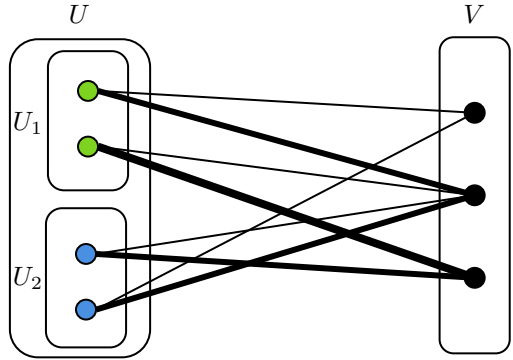
\begin{figure}[t]
	\begin{center}
\begin{tikzpicture}[x=0.75pt,y=0.75pt,yscale=-1,xscale=1]
	
	\draw  [fill={rgb, 255:red, 126; green, 211; blue, 33 }  ,fill opacity=1 ] (229,176) .. controls (229,173.24) and (231.24,171) .. (234,171) .. controls (236.76,171) and (239,173.24) .. (239,176) .. controls (239,178.76) and (236.76,181) .. (234,181) .. controls (231.24,181) and (229,178.76) .. (229,176) -- cycle ;
	\draw  [fill={rgb, 255:red, 74; green, 144; blue, 226 }  ,fill opacity=1 ] (228,258) .. controls (228,255.24) and (230.24,253) .. (233,253) .. controls (235.76,253) and (238,255.24) .. (238,258) .. controls (238,260.76) and (235.76,263) .. (233,263) .. controls (230.24,263) and (228,260.76) .. (228,258) -- cycle ;
	\draw  [fill={rgb, 255:red, 126; green, 211; blue, 33 }  ,fill opacity=1 ] (229,204) .. controls (229,201.24) and (231.24,199) .. (234,199) .. controls (236.76,199) and (239,201.24) .. (239,204) .. controls (239,206.76) and (236.76,209) .. (234,209) .. controls (231.24,209) and (229,206.76) .. (229,204) -- cycle ;
	\draw  [fill={rgb, 255:red, 0; green, 0; blue, 0 }  ,fill opacity=1 ] (422,270) .. controls (422,267.24) and (424.24,265) .. (427,265) .. controls (429.76,265) and (432,267.24) .. (432,270) .. controls (432,272.76) and (429.76,275) .. (427,275) .. controls (424.24,275) and (422,272.76) .. (422,270) -- cycle ;
	\draw  [fill={rgb, 255:red, 0; green, 0; blue, 0 }  ,fill opacity=1 ] (422,228.5) .. controls (422,225.74) and (424.24,223.5) .. (427,223.5) .. controls (429.76,223.5) and (432,225.74) .. (432,228.5) .. controls (432,231.26) and (429.76,233.5) .. (427,233.5) .. controls (424.24,233.5) and (422,231.26) .. (422,228.5) -- cycle ;
	\draw  [fill={rgb, 255:red, 0; green, 0; blue, 0 }  ,fill opacity=1 ] (422,187) .. controls (422,184.24) and (424.24,182) .. (427,182) .. controls (429.76,182) and (432,184.24) .. (432,187) .. controls (432,189.76) and (429.76,192) .. (427,192) .. controls (424.24,192) and (422,189.76) .. (422,187) -- cycle ;
	\draw  [fill={rgb, 255:red, 74; green, 144; blue, 226 }  ,fill opacity=1 ] (228,286) .. controls (228,283.24) and (230.24,281) .. (233,281) .. controls (235.76,281) and (238,283.24) .. (238,286) .. controls (238,288.76) and (235.76,291) .. (233,291) .. controls (230.24,291) and (228,288.76) .. (228,286) -- cycle ;
	\draw   (246,156) .. controls (250.42,156) and (254,159.58) .. (254,164) -- (254,218) .. controls (254,222.42) and (250.42,226) .. (246,226) -- (222,226) .. controls (217.58,226) and (214,222.42) .. (214,218) -- (214,164) .. controls (214,159.58) and (217.58,156) .. (222,156) -- cycle ;
	\draw   (245,235) .. controls (249.42,235) and (253,238.58) .. (253,243) -- (253,297) .. controls (253,301.42) and (249.42,305) .. (245,305) -- (221,305) .. controls (216.58,305) and (213,301.42) .. (213,297) -- (213,243) .. controls (213,238.58) and (216.58,235) .. (221,235) -- cycle ;
	\draw   (251,150) .. controls (258.73,150) and (265,156.27) .. (265,164) -- (265,296) .. controls (265,303.73) and (258.73,310) .. (251,310) -- (209,310) .. controls (201.27,310) and (195,303.73) .. (195,296) -- (195,164) .. controls (195,156.27) and (201.27,150) .. (209,150) -- cycle ;
	\draw   (437.5,149) .. controls (441.37,149) and (444.5,152.13) .. (444.5,156) -- (444.5,301) .. controls (444.5,304.87) and (441.37,308) .. (437.5,308) -- (416.5,308) .. controls (412.63,308) and (409.5,304.87) .. (409.5,301) -- (409.5,156) .. controls (409.5,152.13) and (412.63,149) .. (416.5,149) -- cycle ;
	\draw    (239,176) -- (427,187) ;
	\draw    (239,204) -- (427,228.5) ;
	\draw [line width=2.25]    (239,176) -- (427,228.5) ;
	\draw [line width=2.25]    (238,258) -- (427,270) ;
	\draw [line width=2.25]    (238,286) -- (427,228.5) ;
	\draw    (238,286) -- (427,187) ;
	\draw    (238,258) -- (427,228.5) ;
	\draw [line width=3]    (239,204) -- (427,270) ;
	
	\draw (223,131) node [anchor=north west][inner sep=0.75pt]    {$U$};
	\draw (420,131) node [anchor=north west][inner sep=0.75pt]    {$V$};
	\draw (195,185) node [anchor=north west][inner sep=0.75pt]    {$U_{1}$};
	\draw (195,263) node [anchor=north west][inner sep=0.75pt]    {$U_{2}$};
\end{tikzpicture}
	\end{center}
	\caption{Example of a left-perfect many-to-one fair matching}
\end{figure}

\paragraph*{Our Contributions}
In this paper, we investigate the parameterized complexity of fair matching problem with respect to different structural parameters of the input graph such as vertex cover number ($\vc$), treewidth ($\tw$), path-width ($\pw$), mudular-width ($\mw$), clique-width ($\cw$), tree-depth ($\td$), neighborhood diversity ($\nd$), feedback vertex number ($\fvs$) and feedback edge number ($\fes$) as well as the intrinsic parameters such as the number of colors and the maximum degree of both parts in the input graph. In particular, we prove that the problem is FPT with respect to feedback edge number (Theorem~\ref{thm:fes}), neighborhood diversity (Theorem~\ref{thm:nd}) and thus vertex cover number. It is also FPT with respect to treewidth and $\Delta_V$, combined (Theorem~\ref{thm:twDelta}). On the negative side, by a parameterized reduction from \mcc\, we prove that it is W[1]-hard with respect to the feedback vertex number, tree-depth and $\Delta_U$, combined (Theorem~\ref{thm:fvstdDelta}). Moreover, if we consider the number of colors $|C|$ as a parameter, we can prove that the problem is FPT with respect to tree-depth and $|C|$ (Theorem~\ref{thm:tdC}). We also prove that in this result tree-depth cannot be replaced with path-width, since it is W[1]-hard with respect to path-width, $|C|$ and $\Delta_U$, combined (Theorem~\ref{thm:pwCDelta} by a parameterized reduction from \ubp). A summary of our main results is depicted in Figure~\ref{fig:summary}.  For definition of notions in parameterized complexity, see \cite{cygan}. Also, for the definition of structural parameters of graphs see \cite{lampis}. 

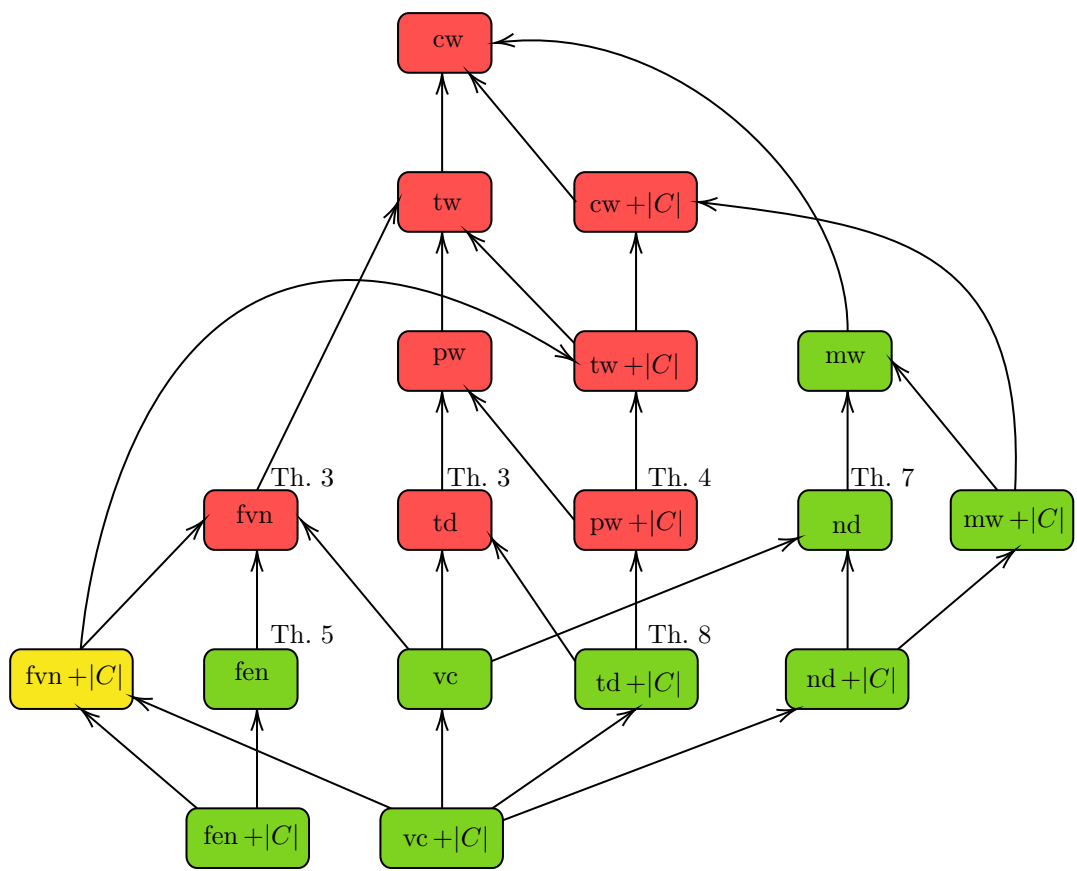
\begin{figure}[ht]
\begin{center}   
\begin{tikzpicture}[x=0.75pt,y=0.75pt,yscale=-1,xscale=.88]
	
	\draw    (275,420) -- (275,372) ;
	\draw [shift={(275,370)}, rotate = 90] [color={rgb, 255:red, 0; green, 0; blue, 0 }  ][line width=0.75]    (10.93,-3.29) .. controls (6.95,-1.4) and (3.31,-0.3) .. (0,0) .. controls (3.31,0.3) and (6.95,1.4) .. (10.93,3.29)   ;
	\draw  [fill={rgb, 255:red, 126; green, 211; blue, 33 }  ,fill opacity=1 ] (240,426) .. controls (240,422.69) and (242.69,420) .. (246,420) -- (303.5,420) .. controls (306.81,420) and (309.5,422.69) .. (309.5,426) -- (309.5,444) .. controls (309.5,447.31) and (306.81,450) .. (303.5,450) -- (246,450) .. controls (242.69,450) and (240,447.31) .. (240,444) -- cycle ;
	\draw  [fill={rgb, 255:red, 126; green, 211; blue, 33 }  ,fill opacity=1 ] (350.5,346) .. controls (350.5,342.69) and (353.19,340) .. (356.5,340) -- (414,340) .. controls (417.31,340) and (420,342.69) .. (420,346) -- (420,364) .. controls (420,367.31) and (417.31,370) .. (414,370) -- (356.5,370) .. controls (353.19,370) and (350.5,367.31) .. (350.5,364) -- cycle ;
	\draw  [fill={rgb, 255:red, 255; green, 80; blue, 80 }  ,fill opacity=1 ] (350,266) .. controls (350,262.69) and (352.69,260) .. (356,260) -- (413.5,260) .. controls (416.81,260) and (419.5,262.69) .. (419.5,266) -- (419.5,284) .. controls (419.5,287.31) and (416.81,290) .. (413.5,290) -- (356,290) .. controls (352.69,290) and (350,287.31) .. (350,284) -- cycle ;
	\draw  [fill={rgb, 255:red, 255; green, 80; blue, 80 }  ,fill opacity=1 ] (350,186) .. controls (350,182.69) and (352.69,180) .. (356,180) -- (413.5,180) .. controls (416.81,180) and (419.5,182.69) .. (419.5,186) -- (419.5,204) .. controls (419.5,207.31) and (416.81,210) .. (413.5,210) -- (356,210) .. controls (352.69,210) and (350,207.31) .. (350,204) -- cycle ;
	\draw  [fill={rgb, 255:red, 255; green, 80; blue, 80 }  ,fill opacity=1 ] (350,106) .. controls (350,102.69) and (352.69,100) .. (356,100) -- (413.5,100) .. controls (416.81,100) and (419.5,102.69) .. (419.5,106) -- (419.5,124) .. controls (419.5,127.31) and (416.81,130) .. (413.5,130) -- (356,130) .. controls (352.69,130) and (350,127.31) .. (350,124) -- cycle ;
	\draw  [fill={rgb, 255:red, 126; green, 211; blue, 33 }  ,fill opacity=1 ] (470,346) .. controls (470,342.69) and (472.69,340) .. (476,340) -- (533.5,340) .. controls (536.81,340) and (539.5,342.69) .. (539.5,346) -- (539.5,364) .. controls (539.5,367.31) and (536.81,370) .. (533.5,370) -- (476,370) .. controls (472.69,370) and (470,367.31) .. (470,364) -- cycle ;
	\draw  [fill={rgb, 255:red, 126; green, 211; blue, 33 }  ,fill opacity=1 ] (130,426) .. controls (130,422.69) and (132.69,420) .. (136,420) -- (193.5,420) .. controls (196.81,420) and (199.5,422.69) .. (199.5,426) -- (199.5,444) .. controls (199.5,447.31) and (196.81,450) .. (193.5,450) -- (136,450) .. controls (132.69,450) and (130,447.31) .. (130,444) -- cycle ;
	\draw  [fill={rgb, 255:red, 248; green, 231; blue, 28 }  ,fill opacity=1 ] (30,346) .. controls (30,342.69) and (32.69,340) .. (36,340) -- (93.5,340) .. controls (96.81,340) and (99.5,342.69) .. (99.5,346) -- (99.5,364) .. controls (99.5,367.31) and (96.81,370) .. (93.5,370) -- (36,370) .. controls (32.69,370) and (30,367.31) .. (30,364) -- cycle ;
	\draw  [fill={rgb, 255:red, 255; green, 80; blue, 80 }  ,fill opacity=1 ] (250,26) .. controls (250,22.69) and (252.69,20) .. (256,20) -- (297,20) .. controls (300.31,20) and (303,22.69) .. (303,26) -- (303,44) .. controls (303,47.31) and (300.31,50) .. (297,50) -- (256,50) .. controls (252.69,50) and (250,47.31) .. (250,44) -- cycle ;
	\draw  [fill={rgb, 255:red, 126; green, 211; blue, 33 }  ,fill opacity=1 ] (477,186) .. controls (477,182.69) and (479.69,180) .. (483,180) -- (524,180) .. controls (527.31,180) and (530,182.69) .. (530,186) -- (530,204) .. controls (530,207.31) and (527.31,210) .. (524,210) -- (483,210) .. controls (479.69,210) and (477,207.31) .. (477,204) -- cycle ;
	\draw  [fill={rgb, 255:red, 126; green, 211; blue, 33 }  ,fill opacity=1 ] (250,346) .. controls (250,342.69) and (252.69,340) .. (256,340) -- (297,340) .. controls (300.31,340) and (303,342.69) .. (303,346) -- (303,364) .. controls (303,367.31) and (300.31,370) .. (297,370) -- (256,370) .. controls (252.69,370) and (250,367.31) .. (250,364) -- cycle ;
	\draw  [fill={rgb, 255:red, 255; green, 80; blue, 80 }  ,fill opacity=1 ] (250,266) .. controls (250,262.69) and (252.69,260) .. (256,260) -- (297,260) .. controls (300.31,260) and (303,262.69) .. (303,266) -- (303,284) .. controls (303,287.31) and (300.31,290) .. (297,290) -- (256,290) .. controls (252.69,290) and (250,287.31) .. (250,284) -- cycle ;
	\draw  [fill={rgb, 255:red, 255; green, 80; blue, 80 }  ,fill opacity=1 ] (250,186) .. controls (250,182.69) and (252.69,180) .. (256,180) -- (297,180) .. controls (300.31,180) and (303,182.69) .. (303,186) -- (303,204) .. controls (303,207.31) and (300.31,210) .. (297,210) -- (256,210) .. controls (252.69,210) and (250,207.31) .. (250,204) -- cycle ;
	\draw  [fill={rgb, 255:red, 255; green, 80; blue, 80 }  ,fill opacity=1 ] (250,106) .. controls (250,102.69) and (252.69,100) .. (256,100) -- (297,100) .. controls (300.31,100) and (303,102.69) .. (303,106) -- (303,124) .. controls (303,127.31) and (300.31,130) .. (297,130) -- (256,130) .. controls (252.69,130) and (250,127.31) .. (250,124) -- cycle ;
	\draw  [fill={rgb, 255:red, 126; green, 211; blue, 33 }  ,fill opacity=1 ] (140,346) .. controls (140,342.69) and (142.69,340) .. (146,340) -- (187,340) .. controls (190.31,340) and (193,342.69) .. (193,346) -- (193,364) .. controls (193,367.31) and (190.31,370) .. (187,370) -- (146,370) .. controls (142.69,370) and (140,367.31) .. (140,364) -- cycle ;
	\draw  [fill={rgb, 255:red, 255; green, 80; blue, 80 }  ,fill opacity=1 ] (140,266) .. controls (140,262.69) and (142.69,260) .. (146,260) -- (187,260) .. controls (190.31,260) and (193,262.69) .. (193,266) -- (193,284) .. controls (193,287.31) and (190.31,290) .. (187,290) -- (146,290) .. controls (142.69,290) and (140,287.31) .. (140,284) -- cycle ;
	\draw  [fill={rgb, 255:red, 126; green, 211; blue, 33 }  ,fill opacity=1 ] (477,266) .. controls (477,262.69) and (479.69,260) .. (483,260) -- (524,260) .. controls (527.31,260) and (530,262.69) .. (530,266) -- (530,284) .. controls (530,287.31) and (527.31,290) .. (524,290) -- (483,290) .. controls (479.69,290) and (477,287.31) .. (477,284) -- cycle ;
	\draw  [fill={rgb, 255:red, 126; green, 211; blue, 33 }  ,fill opacity=1 ] (563.5,266) .. controls (563.5,262.69) and (566.19,260) .. (569.5,260) -- (627,260) .. controls (630.31,260) and (633,262.69) .. (633,266) -- (633,284) .. controls (633,287.31) and (630.31,290) .. (627,290) -- (569.5,290) .. controls (566.19,290) and (563.5,287.31) .. (563.5,284) -- cycle ;
	\draw    (275,340) -- (275,292) ;
	\draw [shift={(275,290)}, rotate = 90] [color={rgb, 255:red, 0; green, 0; blue, 0 }  ][line width=0.75]    (10.93,-3.29) .. controls (6.95,-1.4) and (3.31,-0.3) .. (0,0) .. controls (3.31,0.3) and (6.95,1.4) .. (10.93,3.29)   ;
	\draw    (275,260) -- (275,212) ;
	\draw [shift={(275,210)}, rotate = 90] [color={rgb, 255:red, 0; green, 0; blue, 0 }  ][line width=0.75]    (10.93,-3.29) .. controls (6.95,-1.4) and (3.31,-0.3) .. (0,0) .. controls (3.31,0.3) and (6.95,1.4) .. (10.93,3.29)   ;
	\draw    (275,180) -- (275,132) ;
	\draw [shift={(275,130)}, rotate = 90] [color={rgb, 255:red, 0; green, 0; blue, 0 }  ][line width=0.75]    (10.93,-3.29) .. controls (6.95,-1.4) and (3.31,-0.3) .. (0,0) .. controls (3.31,0.3) and (6.95,1.4) .. (10.93,3.29)   ;
	\draw    (275,100) -- (275,52) ;
	\draw [shift={(275,50)}, rotate = 90] [color={rgb, 255:red, 0; green, 0; blue, 0 }  ][line width=0.75]    (10.93,-3.29) .. controls (6.95,-1.4) and (3.31,-0.3) .. (0,0) .. controls (3.31,0.3) and (6.95,1.4) .. (10.93,3.29)   ;
	\draw    (505,260) -- (505,212) ;
	\draw [shift={(505,210)}, rotate = 90] [color={rgb, 255:red, 0; green, 0; blue, 0 }  ][line width=0.75]    (10.93,-3.29) .. controls (6.95,-1.4) and (3.31,-0.3) .. (0,0) .. controls (3.31,0.3) and (6.95,1.4) .. (10.93,3.29)   ;
	\draw    (505,340) -- (505,292) ;
	\draw [shift={(505,290)}, rotate = 90] [color={rgb, 255:red, 0; green, 0; blue, 0 }  ][line width=0.75]    (10.93,-3.29) .. controls (6.95,-1.4) and (3.31,-0.3) .. (0,0) .. controls (3.31,0.3) and (6.95,1.4) .. (10.93,3.29)   ;
	\draw    (385,180) -- (385,132) ;
	\draw [shift={(385,130)}, rotate = 90] [color={rgb, 255:red, 0; green, 0; blue, 0 }  ][line width=0.75]    (10.93,-3.29) .. controls (6.95,-1.4) and (3.31,-0.3) .. (0,0) .. controls (3.31,0.3) and (6.95,1.4) .. (10.93,3.29)   ;
	\draw    (385,260) -- (385,212) ;
	\draw [shift={(385,210)}, rotate = 90] [color={rgb, 255:red, 0; green, 0; blue, 0 }  ][line width=0.75]    (10.93,-3.29) .. controls (6.95,-1.4) and (3.31,-0.3) .. (0,0) .. controls (3.31,0.3) and (6.95,1.4) .. (10.93,3.29)   ;
	\draw    (385,340) -- (385,292) ;
	\draw [shift={(385,290)}, rotate = 90] [color={rgb, 255:red, 0; green, 0; blue, 0 }  ][line width=0.75]    (10.93,-3.29) .. controls (6.95,-1.4) and (3.31,-0.3) .. (0,0) .. controls (3.31,0.3) and (6.95,1.4) .. (10.93,3.29)   ;
	\draw    (170,340) -- (170,292) ;
	\draw [shift={(170,290)}, rotate = 90] [color={rgb, 255:red, 0; green, 0; blue, 0 }  ][line width=0.75]    (10.93,-3.29) .. controls (6.95,-1.4) and (3.31,-0.3) .. (0,0) .. controls (3.31,0.3) and (6.95,1.4) .. (10.93,3.29)   ;
	\draw    (170,420) -- (170,372) ;
	\draw [shift={(170,370)}, rotate = 90] [color={rgb, 255:red, 0; green, 0; blue, 0 }  ][line width=0.75]    (10.93,-3.29) .. controls (6.95,-1.4) and (3.31,-0.3) .. (0,0) .. controls (3.31,0.3) and (6.95,1.4) .. (10.93,3.29)   ;
	\draw    (136,420) -- (71.59,371.21) ;
	\draw [shift={(70,370)}, rotate = 37.15] [color={rgb, 255:red, 0; green, 0; blue, 0 }  ][line width=0.75]    (10.93,-3.29) .. controls (6.95,-1.4) and (3.31,-0.3) .. (0,0) .. controls (3.31,0.3) and (6.95,1.4) .. (10.93,3.29)   ;
	\draw    (303.5,420) -- (383.3,371.05) ;
	\draw [shift={(385,370)}, rotate = 148.47] [color={rgb, 255:red, 0; green, 0; blue, 0 }  ][line width=0.75]    (10.93,-3.29) .. controls (6.95,-1.4) and (3.31,-0.3) .. (0,0) .. controls (3.31,0.3) and (6.95,1.4) .. (10.93,3.29)   ;
	\draw    (309.5,426) -- (474.1,370.64) ;
	\draw [shift={(476,370)}, rotate = 161.41] [color={rgb, 255:red, 0; green, 0; blue, 0 }  ][line width=0.75]    (10.93,-3.29) .. controls (6.95,-1.4) and (3.31,-0.3) .. (0,0) .. controls (3.31,0.3) and (6.95,1.4) .. (10.93,3.29)   ;
	\draw    (533.5,340) -- (598.4,291.2) ;
	\draw [shift={(600,290)}, rotate = 143.06] [color={rgb, 255:red, 0; green, 0; blue, 0 }  ][line width=0.75]    (10.93,-3.29) .. controls (6.95,-1.4) and (3.31,-0.3) .. (0,0) .. controls (3.31,0.3) and (6.95,1.4) .. (10.93,3.29)   ;
	\draw    (303,346) -- (475.12,284.67) ;
	\draw [shift={(477,284)}, rotate = 160.39] [color={rgb, 255:red, 0; green, 0; blue, 0 }  ][line width=0.75]    (10.93,-3.29) .. controls (6.95,-1.4) and (3.31,-0.3) .. (0,0) .. controls (3.31,0.3) and (6.95,1.4) .. (10.93,3.29)   ;
	\draw    (256,340) -- (196.37,276.46) ;
	\draw [shift={(195,275)}, rotate = 46.82] [color={rgb, 255:red, 0; green, 0; blue, 0 }  ][line width=0.75]    (10.93,-3.29) .. controls (6.95,-1.4) and (3.31,-0.3) .. (0,0) .. controls (3.31,0.3) and (6.95,1.4) .. (10.93,3.29)   ;
	\draw    (246,420) -- (101.37,364.71) ;
	\draw [shift={(99.5,364)}, rotate = 20.92] [color={rgb, 255:red, 0; green, 0; blue, 0 }  ][line width=0.75]    (10.93,-3.29) .. controls (6.95,-1.4) and (3.31,-0.3) .. (0,0) .. controls (3.31,0.3) and (6.95,1.4) .. (10.93,3.29)   ;
	\draw    (70,340) -- (138.53,276.36) ;
	\draw [shift={(140,275)}, rotate = 137.12] [color={rgb, 255:red, 0; green, 0; blue, 0 }  ][line width=0.75]    (10.93,-3.29) .. controls (6.95,-1.4) and (3.31,-0.3) .. (0,0) .. controls (3.31,0.3) and (6.95,1.4) .. (10.93,3.29)   ;
	\draw    (170,260) -- (249.03,116.75) ;
	\draw [shift={(250,115)}, rotate = 118.89] [color={rgb, 255:red, 0; green, 0; blue, 0 }  ][line width=0.75]    (10.93,-3.29) .. controls (6.95,-1.4) and (3.31,-0.3) .. (0,0) .. controls (3.31,0.3) and (6.95,1.4) .. (10.93,3.29)   ;
	\draw    (350,275) -- (291.36,211.47) ;
	\draw [shift={(290,210)}, rotate = 47.29] [color={rgb, 255:red, 0; green, 0; blue, 0 }  ][line width=0.75]    (10.93,-3.29) .. controls (6.95,-1.4) and (3.31,-0.3) .. (0,0) .. controls (3.31,0.3) and (6.95,1.4) .. (10.93,3.29)   ;
	\draw    (350,186) -- (290.47,131.35) ;
	\draw [shift={(289,130)}, rotate = 42.55] [color={rgb, 255:red, 0; green, 0; blue, 0 }  ][line width=0.75]    (10.93,-3.29) .. controls (6.95,-1.4) and (3.31,-0.3) .. (0,0) .. controls (3.31,0.3) and (6.95,1.4) .. (10.93,3.29)   ;
	\draw    (351,115) -- (291.37,51.46) ;
	\draw [shift={(290,50)}, rotate = 46.82] [color={rgb, 255:red, 0; green, 0; blue, 0 }  ][line width=0.75]    (10.93,-3.29) .. controls (6.95,-1.4) and (3.31,-0.3) .. (0,0) .. controls (3.31,0.3) and (6.95,1.4) .. (10.93,3.29)   ;
	\draw    (505,180) .. controls (506,109.35) and (408.98,22.87) .. (306.54,34.81) ;
	\draw [shift={(305,35)}, rotate = 352.81] [color={rgb, 255:red, 0; green, 0; blue, 0 }  ][line width=0.75]    (10.93,-3.29) .. controls (6.95,-1.4) and (3.31,-0.3) .. (0,0) .. controls (3.31,0.3) and (6.95,1.4) .. (10.93,3.29)   ;
	\draw    (70,340) .. controls (75,286) and (120,64) .. (350,195) ;
	\draw [shift={(350,195)}, rotate = 209.66] [color={rgb, 255:red, 0; green, 0; blue, 0 }  ][line width=0.75]    (10.93,-3.29) .. controls (6.95,-1.4) and (3.31,-0.3) .. (0,0) .. controls (3.31,0.3) and (6.95,1.4) .. (10.93,3.29)   ;
	\draw    (600,260) .. controls (606.97,136.62) and (532.75,127.09) .. (421.68,115.18) ;
	\draw [shift={(420,115)}, rotate = 6.12] [color={rgb, 255:red, 0; green, 0; blue, 0 }  ][line width=0.75]    (10.93,-3.29) .. controls (6.95,-1.4) and (3.31,-0.3) .. (0,0) .. controls (3.31,0.3) and (6.95,1.4) .. (10.93,3.29)   ;
	\draw    (350.5,346) -- (304.22,285.59) ;
	\draw [shift={(303,284)}, rotate = 52.54] [color={rgb, 255:red, 0; green, 0; blue, 0 }  ][line width=0.75]    (10.93,-3.29) .. controls (6.95,-1.4) and (3.31,-0.3) .. (0,0) .. controls (3.31,0.3) and (6.95,1.4) .. (10.93,3.29)   ;
	\draw    (590,260) -- (531.36,196.47) ;
	\draw [shift={(530,195)}, rotate = 47.29] [color={rgb, 255:red, 0; green, 0; blue, 0 }  ][line width=0.75]    (10.93,-3.29) .. controls (6.95,-1.4) and (3.31,-0.3) .. (0,0) .. controls (3.31,0.3) and (6.95,1.4) .. (10.93,3.29)   ;
	
	\draw (268,350) node [anchor=north west][inner sep=0.75pt]    {$\vc$};
	\draw (268,187) node [anchor=north west][inner sep=0.75pt]    {$\pw$};
	\draw (494.5,271) node [anchor=north west][inner sep=0.75pt]    {$\nd$};
	\draw (268,270) node [anchor=north west][inner sep=0.75pt]    {$\td$};
	\draw (268,30) node [anchor=north west][inner sep=0.75pt]    {$\cw$};
	\draw (268,109) node [anchor=north west][inner sep=0.75pt]    {$\tw$};
	\draw (156.5,266) node [anchor=north west][inner sep=0.75pt]    {$\fvs$};
	\draw (490.5,189) node [anchor=north west][inner sep=0.75pt]    {$\mw$};
	\draw (155.5,345) node [anchor=north west][inner sep=0.75pt]    {$\fes$};
	\draw (251,427.5) node [anchor=north west][inner sep=0.75pt]    {$\vc+|C|$};
	\draw (480,347.5) node [anchor=north west][inner sep=0.75pt]    {$\nd+|C|$};
	\draw (569,267) node [anchor=north west][inner sep=0.75pt]    {$\mw+|C|$};
	\draw (360.5,349.5) node [anchor=north west][inner sep=0.75pt]    {$\td+|C|$};
	\draw (357,268.5) node [anchor=north west][inner sep=0.75pt]    {$\pw+|C|$};
	\draw (357,189.5) node [anchor=north west][inner sep=0.75pt]    {$\tw+|C|$};
	\draw (357,108.5) node [anchor=north west][inner sep=0.75pt]    {$\cw+|C|$};
	\draw (138,426.5) node [anchor=north west][inner sep=0.75pt]    {$\fes+|C|$};
	\draw (37,346) node [anchor=north west][inner sep=0.75pt]    {$\fvs+|C|$};
	\draw (176,326) node [anchor=north west][inner sep=0.75pt]    {Th.~\ref{thm:fes}};
	\draw (176,247) node [anchor=north west][inner sep=0.75pt]    {Th.~\ref{thm:fvstdDelta}};
	\draw (276,247) node [anchor=north west][inner sep=0.75pt]    {Th.~\ref{thm:fvstdDelta}};
	\draw (390,247) node [anchor=north west][inner sep=0.75pt]    {Th.~\ref{thm:pwCDelta}};
	\draw (390,326) node [anchor=north west][inner sep=0.75pt]    {Th.~\ref{thm:tdC}};
	\draw (505,247) node [anchor=north west][inner sep=0.75pt]    {Th.~\ref{thm:nd}};

\end{tikzpicture}
	\end{center}
	\caption{Summary of our main results regarding \gfm\ problem. An arrow from $f$ to $g$ means that $g$ is bounded by a function of $f$ and so W[1]-hardness result with respect to $f$ implies W[1]-hardness with respect to $g$. Parameters marked by green are proved to be FPT. Parameters marked by red are proved to be W[1]-hard and the situation of the parameter marked by yellow is unknown.} \label{fig:summary}
\end{figure}
\paragraph*{Related Work}
The first study of a fairness notion in combinatorial problems  dates back to 1989 in \cite{lin} where the \textsc{Fair Edge Deletion} problem is studied. In this problem, we are going to delete some edges of the graph to obtain an acyclic subgraph such that the number of deleted edges incident with each vertex is bounded by a given constant.
In general, a \textsc{Fair Deletion Problem} seeks for a set of elements (vertices or edges) $S$ of a graph $G$ such that the subgraph obtained by deletion of $S$ from $G$ satisfies  a given property and the number of deleted elements in the neighborhood of each vertex is minimized. The parameterized complexity of these problems is investigated in  \cite{knop,masarik,kolman}.
Two special problems in the class of fair vertex deletion problems are \textsc{Fair Vertex Cover} and \textsc{Fair Feedback Vertex Set} problems in which we are seeking for a vertex cover (resp. a feedback vertex set) $S$ such that the number of neighbors of every vertex in $S$ is bounded by a given constant. In \cite{knop} it is proved that \textsc{Fair Vertex Cover} is W[1]-hard with respect to tree-depth and feedback vertex number of the input graph and FPT with respect to modular-width of the input graph. 
In \cite{kanesh}, it is proved that \textsc{Fair Feedback Vertex Set} is W[1]-hard with respect to tree-depth and FPT with respect to neighborhood diversity of the input graph. 

A generalization of \textsc{Fair Vertex Cover} problem is \textsc{Fair Hitting Set} in which we are given a universe $\mathcal{U}$ and two families $\mathcal{F}$ and $\mathcal{B}$ of subsets of $\mathcal{U}$ and we are looking for a subset $S\subseteq \mathcal{U}$ of a given size $k$ which hits every element of $\mathcal{F}$ and its intersection with each $B\in \mathcal{B}$ is bounded by a given constant. In \cite{inamdar}, among other results, it is proved that the problem is W[1]-hard with respect to $k$ and if every element of $\mathcal{U}$ appears in at most $q$ sets in $\mathcal{B}$ and $d$ sets in $\mathcal{F}$, then it is FPT with respect to $q, d$ and $k$. 

The Max-Min fairness measure is also considered in \textsc{Fair Short Path} problem \cite{bentert} in which we are given a graph $G$ whose vertex set is colored, two vertices $s,t$ and two integers $\ell, \delta$ and the question is that if there exists an $s-t-$path $P$ of length at most $\ell$ such that the difference of the numbers of vertices in $P$ in each two color classes is bounded by $\delta$. Bentert et al. \cite{bentert} proved that \textsc{Fair Short Path} is in XP and W[1]-hard with respect to the number of colors and is FPT with respect to $\ell$ and $\delta$. 

The margin of victory fairness measure (MoV) is studied e.g. in \textsc{Fair Connected Districting} problem in which we are given a vertex-colored graph and we are looking for a partition of vertices into $k$ districts (connected subgraphs) such that each district is $\ell$-fair with MoV measure. This problem is firstly introduced by Stoica et al. \cite{stoica} and, the parameterized complexity of the problem is investigated in \cite{bomer} In particular, it is proved that the problem is W[1]-hard with respect to treewidth, $k$ and the number of colors, combined. Also, it is W[1]-hard with respect to feedback edge number and $k$, combined. It is also FPT with respect to vertex cover number and the number of colors, combined. 

The notion of fairness is also studied widely in the literature of machine learning such as data clustering \cite{friggstad,ahmadian,ahmadian2020,froese}, matroids and matchings \cite{chierichetti},  influence maximization \cite{khajehnejad} and graph mining \cite{kang,dong}. 

\section{Preliminaries}
For positive integers $m,n,m < n $, let us define $[m,n]=\{m,m+1,\dots,n\}$ and $[n]=[1,n]$. For a graph $G=(V,E)$ and two subsets $A,B \subseteq E$ the set $E(A,B)$ denotes the set of all edges in $E$ with one endpoint in $A$ and one endpoint in $B$. Also, for a vertex $v \in V$, we say that $v$ is complete (resp. incomplete) to $A$ if $v$ is adjacent (resp. nonadjacent) to all vertices in $A$.

 Let $G=(U\cup V,E) $ be a bipartite graph. For an edge set $M\subseteq E$ and a vertex $x\in U\cup V$, let us define $M(x)$ be the set of all vertices $y$ such that $xy$ is in $M$, i.e. $M(x)=\{y\in U\cup V: xy\in M\} $. We say that $M$ is a \textit{many-to-one matching} if for every vertex $u\in U$, we have $|M(u)|\leq 1$. A many-to-one matching is called \textit{left-perfect} if $|M(u)|=1$ for all $u\in U$.

Now, let $C$ be a set of colors and $\col:U\to C$ be a coloring of vertices in $U$.  For a set $U'\subseteq U$ and a color $c\in C$,  we define $U'_c$ as the set of all vertices in $U'$ with color $c$. Also, define $G_c=G[U_c\cup V]$.  Now, let $\ell$ be an integer. we say that $U'$ is $\ell$-fair if  $\max_{c\in C} |U'_c|-\min_{c\in C} |U'_c|\leq \ell  $.
A many-to-one matching is called $\ell$-fair if for every vertex $v\in V$, $M(v)$ is $\ell$-fair. Now, we define \textsc{Fair Matching} problem as follows.

\prob{Fair Matching}{A bipartite graph $G=(U\cup V,E)$, a coloring $\col:U\to C$ and a non-negative integer $\ell$.}{Is there a left-perfect many-to-one matching for $G$ which is $\ell$-fair?}

We can also generalize the concept of $\ell$-fairness as follows. Suppose that $L:V\to \mathbb{N}\cup \{0\}$ is a threshold function which assigns an integer to each vertex in $V$. Then, a many-to-one matching is called $L$-fair if for every vertex $v\in V$, $M(v)$ is $L(v)$-fair. 
\textsc{generalized fair matching} can be defined as follows.

\prob{Generalized Fair Matching}{A bipartite graph $G=(U\cup V,E)$, a set of colors $C$, a coloring $\col:U\to C$ and a threshold function $L:V\to \mathbb{N}\cup \{0\}$.}{Is there a left-perfect many-to-one matching for $G$ which is $L$-fair?}

In \cite{main}, it is proved that \textsc{Fair Matching} in FPT with respect to $k=|V|$. In fact, they proposed the following ILP formulation for the problem.

Let $M$ be a many-to-one matching and for each vertex $v\in V$, let $x_v$ (resp. $y_v$) be the minimum (resp. maximum) number of vertices of some color $c\in C$ matched to $v$, i.e. $x_v =\max_{c\in C} M(v)_c$, $y_v= \min_{c\in C} M(v)_c$. For a vertex $w\in U\cup V$, let $N_G(w)$ be the set of neighbors of $w$ in $G$. For each $W\subseteq U\cup V$, let $N_G(W)=\cup_{w\in W} N_G(w)$. Also, define $\nu_G(W)=\{w'\in U\cup V: N_G(w')\subseteq W\}$. Moreover, for each $c\in C$, let us define $N_c(W)=N_{G_c}(W)$ and $\nu_c(W)=N_{G_c}(W)$. then, we can define the following ILP with variables $x_v , y_v , v \in V$.
\begin{align}
&\text{ILP1: } \nonumber \\
&\sum_{v\in W} y_v\geq \max_{c\in C} |\nu_c(W)| \quad & \forall\ W\subseteq V\\
&\sum_{v\in W} x_v\leq \min_{c\in C} |N_c(W)| \quad  &\forall\ W\subseteq V\\
&0\leq y_v-x_v\leq \ell \quad& \forall v\in V \label{con:fair}\\
&x_v,y_v\in \mathbb{Z}^+\cup\{0\} \quad& \forall v\in V
\end{align}
It is proved in \cite{main} that ILP1 has a feasible solution if and only if the graph G has a left-perfect many-to-one matching which is $\ell$-fair. ILP1 has $O(k)$ variables and $O(2^k)$ constraints. So, by a result of Lenstra \cite{lenstra}, it can be solved in time $O^*(k^{O(k)})$ and so Fair Matching is FPT with respect to $k$. It can be easily seen that in Condition~\eqref{con:fair}, $\ell$ can be replaced with $L(v)$ and so this result is also valid for \textsc{Generalized Fair Matching}. 

\begin{theorem}\label{thm:FPTk}
Let $G=(U\cup V, E)$ be the input graph and $|V|=k$. Then, \textsc{Generalized Fair Matching} can be solved in time $O^*(k^{O(k)})$. 
\end{theorem}

In order to prove that Generalized Fair Matching is FPT with respect to ($\td , |C|$), we will write another ILP formulation for this problem (see ILP2) which has many more number of variables but the structure of the associated graph to this ILP is related to the structure of the input graph. For this, we need some results in parameterized complexity of integer linear programming (ILP). Consider a standard form of ILP as follows. 
\begin{equation}\label{eq:ILP}
\min \{\mathbf{w}.\mathbf{x}\mid A\mathbf{x}=\mathbf{b},\ \mathbf{l}\leq \mathbf{x}\leq \mathbf{u},\ \mathbf{x}\in \mathbb{Z}^n\},
\end{equation}
where $A$ is an integer $m\times n$ matrix, $\mathbf{b}\in \mathbb{Z}^m$ and $\mathbf{l,u}\in (\mathbb{Z}\cup\{\pm \infty\})^n$.
It is known that solving an ILP in its general form is NP-hard, however it can be solved in FPT time when the structure of the coefficient matrix $A$ is confined with a special parameter. To do this, we associate to $A$ two graphs, namely \textit{primal} and \textit{dual} graphs as follows. The primal graph $G_P(A)$ is the graph whose vertex set is corresponding to the columns of $A$ and two vertices corresponding to columns $j_1,j_2$ are adjacent if there is a row $i$ such that $A_{ij_1}$ and $A_{ij_2}$ are non-zero. Similarly, the dual graph $G_D(A)$ is defined as the graph with vertex set corresponding to the rows of $A$ and two vertices corresponding to rows $i_1,i_2$ are adjacent if there is a column $j$ such that $A_{i_1j}$ and $A_{i_2j}$ are non-zero. The primal and dual tree-depth of $A$ denoted by $\td_P(A)$ and $\td_D(A)$ respectively is defined as the tree-depth of the primal and dual graph, respectively. We need the following results which asserts that ILP can be solved in FPT time with respect to the dual tree-depth of $A$ and $\|A\|_\infty$.
\begin{theorem} {\rm \cite{eisenbrand}} \label{thm:tdDual}
	The ILP in \eqref{eq:ILP} can be solved in time $O^*((\|A\|_\infty+1)^{2^{\td_D(A)}})$.
\end{theorem}    
\section{W[1]-hard Results}

In this section, we prove two W[1]-hard results regarding \gfm. In the first result the parameters are feedback vertex number, tree-depth and $\Delta_U$, combined and reduction is from \mcc. In the second result, the parameters are path-with, number of colors and $\Delta_U$, combined and the reduction is from \ubp.

\begin{theorem} \label{thm:fvstdDelta}
\gfm\ is W[1]-hard with respect to $(\fvs,\td)$ even on graphs with $\Delta_U=2$ where $\fvs$ and $\td$ are  the feedback vertex number and the tree-depth of the input graph and $\Delta_U$ is the maximum degree of vertices in $U$. 
\end{theorem}
\begin{proof}
We give a parameterized reduction from \mcc\ which is known to be W[1]-hard with respect to the size of the clique $\ell$ \cite{fellow}. 

\prob{Multicolored Clique}{A graph $G=(V,E)$ where the vertex set $V$ is partitioned into $\ell$ independent sets $U_1,\ldots,U_{\ell}$.}{Is there a clique of size $\ell$ in $G$?}

Let $G=(U_1\cup\cdots\cup U_{\ell},E)$ be an instance of \mcc\ such that for each $a\in [\ell]$, $|U_a|=n$. We define an ordering on the vertices in each $U_a$ and if there is an edge between $i$-th vertex in $U_a$ and $j$-th vertex in $U_b$, we denote such edge by $ij\in E(U_a,U_b)$. Now, we construct an instance $G'=(U'\cup V',E')$ of \gfm\ whose feedback vertex number and tree-depth are $O(\ell^2)$ and $\Delta_U=2$. The set of colors is defined as $C=\{(a,i): a\in [\ell], i\in [0,n] \}$, so $|C|=(n+1)\ell$.

Firstly, for each edge $ij\in E(U_a,U_b)$, we construct the gadget $H^{ab}_{ij}$ which is defined as follows. The bipartite graph $H^{ab}_{ij}$ has a bipartition $(U,V)$ where $V=\{v_1,v_2,v_3\}$ and $U=U^{ab}_{ij}\cup U^{ba}_{ji}$, where $U^{ab}_{ij}$ is a set of $n-i$ vertices with colors $(b,i+1),\ldots, (b,n)$ and $U^{ba}_{ji}$ is a set of $n-j$ vertices with colors $(a,j+1),\ldots, (a,n)$. The set $U^{ab}_{ij}$ is complete to $\{v_1,v_2\}$ and $U^{ba}_{ji}$ is complete to $\{v_2,v_3\}$. There is no more edges in $H^{ab}_{ij}$. 

Secondly, For every subset of colors $\hat{C}\subseteq C$, we define the gadget $H[\hat{C}]$ as follows. The bipartite graph $H[\hat{C}]$ has a bipartition $(\hat{U},\hat{V})$ where $|\hat{U}|=|\hat{C}|$ and for each color $c\in \hat{C}$, there is exactly one vertex $u^c$ in $ \hat{U} $ of color $c$. There is a vertex $\bar{v}\in \hat{V}$ which is complete $\hat{U}$. Also, each vertex $u^c$ in $ \hat{U}$ has a private neighbor $v^c$ in $\hat{V}$. Also, $L(\bar{v})=0$ and $L(v^c)=1$ for all $c\in \hat{C}$. There is no more edges in $H[\hat{C}]$. 

Now, we construct the graph $H$ as follows. The bipartite graph $H$ has a bipartition $U_H\cup V_H$, where $V_H=\{v_0,v_{ab}, v'_{ai}; a,b\in[\ell],a\neq b, i\in [n] \}$ and 
$$U_H=\bigcup_{a\in [0,\ell]} U_a \bigcup_{a,b\in [\ell]} U_{a,b} \bigcup_{a,b\in [\ell],i\in [n]} U_{a,b,i}. $$
For each $a,b\in [\ell], a\neq b$ and $i\in[n]$, $U_a$ is a set of $n$ vertices all with color $(a,0)$, $U_{a,b,i}$ is a set of $i$ vertices of colors $(b,1),\ldots ,(b,i)$, $U_0$ is a set of $n\ell$ vertices with colors $(a,i)$, $a \in [\ell]$, $i \in [n]$, and finally, $U_{a,b}$ is a set of $(n+1)(\ell-1)$ vertices with colors $(c,i)$, $c \in [\ell] \setminus \{b\}$, $i \in [0,n]$. Adjacencies are as follows. The vertex $v_0$ is complete to $U_0\cup U_1\cup \cdots \cup U_\ell$. $v_{ab}$ is complete to $U_{ab} \bigcup_{i \in [n]} U_{a,b,i}$ and $v'_{ai}$ is complete to $\bigcup_{b \in [\ell]} U_{a,b,i}$ as well as i-th vertex in $U_a$. 

Finally, the bipartite graph $G'=(U\cup V,E')$ is constructed as follows. We take a copy of $H$ and for each edge $ij\in E(U_a,U_b)$, we add a disjoint gadget $H^{ab}_{ij}$ and then identify the vertex $v_1$ with the vertex $v_{ab}$ and the vertex $v_3$ with the vertex $v_{ba}$. Also, for every vertex $v\in V\setminus \{v_0, v_{ab},a,b\in [\ell], a\neq b \}$, setting $\hat{C}_v$ to be the set of colors in $C$ not appearing in the neighbors of $v$, we add a disjoint gadget $H[\hat{C}_v]$ and identify $\bar{v}$ with $v$ (see Figure~\ref{fig:fvs}). Note that for every vertex $v\in V$ except vertices in the gadgets $H[\hat{C}]$, we have $L(v)=0$.
\begin{figure}[ht]
	\begin{center}

\tikzset{every picture/.style={line width=0.75pt}} 

\begin{tikzpicture}[x=0.75pt,y=0.75pt,yscale=-1,xscale=1]
	
	\draw  [fill={rgb, 255:red, 0; green, 0; blue, 0 }  ,fill opacity=1 ] (345,211) .. controls (345,208.24) and (347.24,206) .. (350,206) .. controls (352.76,206) and (355,208.24) .. (355,211) .. controls (355,213.76) and (352.76,216) .. (350,216) .. controls (347.24,216) and (345,213.76) .. (345,211) -- cycle ;
	\draw    (350,211) -- (207.51,270) ;
	\draw  [fill={rgb, 255:red, 0; green, 0; blue, 0 }  ,fill opacity=1 ] (192,376) .. controls (192,373.24) and (194.24,371) .. (197,371) .. controls (199.76,371) and (202,373.24) .. (202,376) .. controls (202,378.76) and (199.76,381) .. (197,381) .. controls (194.24,381) and (192,378.76) .. (192,376) -- cycle ;
	\draw    (197,376) -- (162,441.76) ;
	\draw  [color={rgb, 255:red, 0; green, 0; blue, 0 }  ,draw opacity=1 ] (90,378) .. controls (90,366.95) and (105.67,358) .. (125,358) .. controls (144.33,358) and (160,366.95) .. (160,378) .. controls (160,389.05) and (144.33,398) .. (125,398) .. controls (105.67,398) and (90,389.05) .. (90,378) -- cycle ;
	\draw  [fill={rgb, 255:red, 0; green, 0; blue, 0 }  ,fill opacity=1 ] (457.77,379.52) .. controls (457.77,376.76) and (460.01,374.52) .. (462.77,374.52) .. controls (465.53,374.52) and (467.77,376.76) .. (467.77,379.52) .. controls (467.77,382.28) and (465.53,384.52) .. (462.77,384.52) .. controls (460.01,384.52) and (457.77,382.28) .. (457.77,379.52) -- cycle ;
	\draw    (160,378) -- (197,376) ;
	\draw  [color={rgb, 255:red, 0; green, 0; blue, 0 }  ,draw opacity=1 ] (500,394.14) .. controls (500,383.1) and (515.67,374.14) .. (535,374.14) .. controls (554.33,374.14) and (570,383.1) .. (570,394.14) .. controls (570,405.19) and (554.33,414.14) .. (535,414.14) .. controls (515.67,414.14) and (500,405.19) .. (500,394.14) -- cycle ;
	\draw    (462.77,379.52) -- (500,394.14) ;
	\draw    (386,167) -- (350,211) ;
	\draw  [color={rgb, 255:red, 0; green, 0; blue, 0 }  ,draw opacity=1 ][fill={rgb, 255:red, 238; green, 241; blue, 241 }  ,fill opacity=1 ] (390,306.83) .. controls (390,291.74) and (422.46,279.51) .. (462.49,279.51) .. controls (502.53,279.51) and (534.99,291.74) .. (534.99,306.83) .. controls (534.99,321.91) and (502.53,334.14) .. (462.49,334.14) .. controls (422.46,334.14) and (390,321.91) .. (390,306.83) -- cycle ;
	\draw    (350,211) -- (462.49,279.51) ;
	\draw   (449.51,310.7) .. controls (449.51,304.63) and (454.43,299.7) .. (460.51,299.7) .. controls (466.58,299.7) and (471.51,304.63) .. (471.51,310.7) .. controls (471.51,316.78) and (466.58,321.7) .. (460.51,321.7) .. controls (454.43,321.7) and (449.51,316.78) .. (449.51,310.7) -- cycle ;
	\draw   (493,311.14) .. controls (493,305.07) and (497.92,300.14) .. (504,300.14) .. controls (510.08,300.14) and (515,305.07) .. (515,311.14) .. controls (515,317.22) and (510.08,322.14) .. (504,322.14) .. controls (497.92,322.14) and (493,317.22) .. (493,311.14) -- cycle ;
	\draw   (404.51,310.7) .. controls (404.51,304.63) and (409.43,299.7) .. (415.51,299.7) .. controls (421.58,299.7) and (426.51,304.63) .. (426.51,310.7) .. controls (426.51,316.78) and (421.58,321.7) .. (415.51,321.7) .. controls (409.43,321.7) and (404.51,316.78) .. (404.51,310.7) -- cycle ;
	\draw    (460.51,321.7) -- (462.77,379.52) ;
	\draw    (230,441.76) -- (199,503.24) ;
	\draw    (162,441.76) -- (199,503.24) ;
	\draw  [color={rgb, 255:red, 0; green, 0; blue, 0 }  ,draw opacity=1 ][fill={rgb, 255:red, 238; green, 241; blue, 241 }  ,fill opacity=1 ] (162,441.76) .. controls (162,430.85) and (177.22,422) .. (196,422) .. controls (214.78,422) and (230,430.85) .. (230,441.76) .. controls (230,452.67) and (214.78,461.52) .. (196,461.52) .. controls (177.22,461.52) and (162,452.67) .. (162,441.76) -- cycle ;
	\draw    (197,376) -- (230,440) ;
	\draw  [fill={rgb, 255:red, 0; green, 0; blue, 0 }  ,fill opacity=1 ] (194,503.24) .. controls (194,500.48) and (196.24,498.24) .. (199,498.24) .. controls (201.76,498.24) and (204,500.48) .. (204,503.24) .. controls (204,506) and (201.76,508.24) .. (199,508.24) .. controls (196.24,508.24) and (194,506) .. (194,503.24) -- cycle ;
	\draw    (462.77,379.52) -- (429,444.9) ;
	\draw    (497,444.9) -- (466,506.38) ;
	\draw    (429,444.9) -- (466,506.38) ;
	\draw  [color={rgb, 255:red, 0; green, 0; blue, 0 }  ,draw opacity=1 ][fill={rgb, 255:red, 238; green, 241; blue, 241 }  ,fill opacity=1 ] (429,444.9) .. controls (429,433.99) and (444.22,425.14) .. (463,425.14) .. controls (481.78,425.14) and (497,433.99) .. (497,444.9) .. controls (497,455.81) and (481.78,464.66) .. (463,464.66) .. controls (444.22,464.66) and (429,455.81) .. (429,444.9) -- cycle ;
	\draw    (462.77,379.52) -- (497,444.9) ;
	\draw  [fill={rgb, 255:red, 0; green, 0; blue, 0 }  ,fill opacity=1 ] (461,506.38) .. controls (461,503.62) and (463.24,501.38) .. (466,501.38) .. controls (468.76,501.38) and (471,503.62) .. (471,506.38) .. controls (471,509.14) and (468.76,511.38) .. (466,511.38) .. controls (463.24,511.38) and (461,509.14) .. (461,506.38) -- cycle ;
	\draw    (199,503.24) -- (275,490.48) ;
	\draw    (466,506.38) -- (391,550) ;
	\draw  [fill={rgb, 255:red, 0; green, 0; blue, 0 }  ,fill opacity=1 ] (381,167) .. controls (381,164.24) and (383.24,162) .. (386,162) .. controls (388.76,162) and (391,164.24) .. (391,167) .. controls (391,169.76) and (388.76,172) .. (386,172) .. controls (383.24,172) and (381,169.76) .. (381,167) -- cycle ;
	\draw  [fill={rgb, 255:red, 0; green, 0; blue, 0 }  ,fill opacity=1 ] (311,166) .. controls (311,163.24) and (313.24,161) .. (316,161) .. controls (318.76,161) and (321,163.24) .. (321,166) .. controls (321,168.76) and (318.76,171) .. (316,171) .. controls (313.24,171) and (311,168.76) .. (311,166) -- cycle ;
	\draw    (350,211) -- (316,166) ;
	\draw    (199,503.24) -- (270,550) ;
	\draw  [color={rgb, 255:red, 0; green, 0; blue, 0 }  ,draw opacity=1 ][fill={rgb, 255:red, 238; green, 241; blue, 241 }  ,fill opacity=1 ] (250,520.24) .. controls (250,503.81) and (258.95,490.48) .. (270,490.48) .. controls (281.05,490.48) and (290,503.81) .. (290,520.24) .. controls (290,536.68) and (281.05,550) .. (270,550) .. controls (258.95,550) and (250,536.68) .. (250,520.24) -- cycle ;
	\draw  [color={rgb, 255:red, 0; green, 0; blue, 0 }  ,draw opacity=1 ][fill={rgb, 255:red, 238; green, 241; blue, 241 }  ,fill opacity=1 ] (370,520.24) .. controls (370,503.81) and (379.4,490.48) .. (391,490.48) .. controls (402.6,490.48) and (412,503.81) .. (412,520.24) .. controls (412,536.68) and (402.6,550) .. (391,550) .. controls (379.4,550) and (370,536.68) .. (370,520.24) -- cycle ;
	\draw  [fill={rgb, 255:red, 0; green, 0; blue, 0 }  ,fill opacity=1 ] (328,523.24) .. controls (328,520.48) and (330.24,518.24) .. (333,518.24) .. controls (335.76,518.24) and (338,520.48) .. (338,523.24) .. controls (338,526) and (335.76,528.24) .. (333,528.24) .. controls (330.24,528.24) and (328,526) .. (328,523.24) -- cycle ;
	\draw    (270,550) -- (333,523.24) ;
	\draw    (275,490.48) -- (333,523.24) ;
	\draw    (391,490.48) -- (466,506.38) ;
	\draw    (333,523.24) -- (390,490) ;
	\draw    (333,523.24) -- (391,550) ;
	\draw  [color={rgb, 255:red, 0; green, 0; blue, 0 }  ,draw opacity=1 ][fill={rgb, 255:red, 238; green, 241; blue, 241 }  ,fill opacity=1 ] (135.01,297.31) .. controls (135.01,282.23) and (167.47,270) .. (207.51,270) .. controls (247.54,270) and (280,282.23) .. (280,297.31) .. controls (280,312.4) and (247.54,324.63) .. (207.51,324.63) .. controls (167.47,324.63) and (135.01,312.4) .. (135.01,297.31) -- cycle ;
	\draw   (194.52,301.19) .. controls (194.52,295.12) and (199.45,290.19) .. (205.52,290.19) .. controls (211.6,290.19) and (216.52,295.12) .. (216.52,301.19) .. controls (216.52,307.27) and (211.6,312.19) .. (205.52,312.19) .. controls (199.45,312.19) and (194.52,307.27) .. (194.52,301.19) -- cycle ;
	\draw   (238.01,301.63) .. controls (238.01,295.55) and (242.94,290.63) .. (249.01,290.63) .. controls (255.09,290.63) and (260.01,295.55) .. (260.01,301.63) .. controls (260.01,307.7) and (255.09,312.63) .. (249.01,312.63) .. controls (242.94,312.63) and (238.01,307.7) .. (238.01,301.63) -- cycle ;
	\draw   (149.52,301.19) .. controls (149.52,295.12) and (154.45,290.19) .. (160.52,290.19) .. controls (166.6,290.19) and (171.52,295.12) .. (171.52,301.19) .. controls (171.52,307.27) and (166.6,312.19) .. (160.52,312.19) .. controls (154.45,312.19) and (149.52,307.27) .. (149.52,301.19) -- cycle ;
	\draw    (205.52,312.19) -- (197,376) ;
	\draw  [color={rgb, 255:red, 0; green, 0; blue, 0 }  ,draw opacity=1 ] (300,470) .. controls (300,458.95) and (315.67,450) .. (335,450) .. controls (354.33,450) and (370,458.95) .. (370,470) .. controls (370,481.05) and (354.33,490) .. (335,490) .. controls (315.67,490) and (300,481.05) .. (300,470) -- cycle ;
	\draw    (335,490) -- (333,523.24) ;
	\draw  [fill={rgb, 255:red, 0; green, 0; blue, 0 }  ,fill opacity=1 ] (509,478) .. controls (509,475.24) and (511.24,473) .. (514,473) .. controls (516.76,473) and (519,475.24) .. (519,478) .. controls (519,480.76) and (516.76,483) .. (514,483) .. controls (511.24,483) and (509,480.76) .. (509,478) -- cycle ;
	\draw  [fill={rgb, 255:red, 0; green, 0; blue, 0 }  ,fill opacity=1 ] (509,525) .. controls (509,522.24) and (511.24,520) .. (514,520) .. controls (516.76,520) and (519,522.24) .. (519,525) .. controls (519,527.76) and (516.76,530) .. (514,530) .. controls (511.24,530) and (509,527.76) .. (509,525) -- cycle ;
	\draw    (199,503.24) -- (145,488) ;
	\draw  [fill={rgb, 255:red, 0; green, 0; blue, 0 }  ,fill opacity=1 ] (140,488) .. controls (140,485.24) and (142.24,483) .. (145,483) .. controls (147.76,483) and (150,485.24) .. (150,488) .. controls (150,490.76) and (147.76,493) .. (145,493) .. controls (142.24,493) and (140,490.76) .. (140,488) -- cycle ;
	\draw  [fill={rgb, 255:red, 0; green, 0; blue, 0 }  ,fill opacity=1 ] (139,536) .. controls (139,533.24) and (141.24,531) .. (144,531) .. controls (146.76,531) and (149,533.24) .. (149,536) .. controls (149,538.76) and (146.76,541) .. (144,541) .. controls (141.24,541) and (139,538.76) .. (139,536) -- cycle ;
	\draw    (199,503.24) -- (144,536) ;
	\draw    (466,506.38) -- (514,525) ;
	\draw    (514,478) -- (466,506.38) ;
	\draw  [color={rgb, 255:red, 0; green, 0; blue, 0 }  ,draw opacity=1 ] (285,167.5) .. controls (285,155.07) and (312.98,145) .. (347.5,145) .. controls (382.02,145) and (410,155.07) .. (410,167.5) .. controls (410,179.93) and (382.02,190) .. (347.5,190) .. controls (312.98,190) and (285,179.93) .. (285,167.5) -- cycle ;
	\draw  [color={rgb, 255:red, 0; green, 0; blue, 0 }  ,draw opacity=1 ] (494,500) .. controls (494,477.91) and (502.95,460) .. (514,460) .. controls (525.05,460) and (534,477.91) .. (534,500) .. controls (534,522.09) and (525.05,540) .. (514,540) .. controls (502.95,540) and (494,522.09) .. (494,500) -- cycle ;
	\draw  [color={rgb, 255:red, 0; green, 0; blue, 0 }  ,draw opacity=1 ] (120,510.5) .. controls (120,488.68) and (129.85,471) .. (142,471) .. controls (154.15,471) and (164,488.68) .. (164,510.5) .. controls (164,532.32) and (154.15,550) .. (142,550) .. controls (129.85,550) and (120,532.32) .. (120,510.5) -- cycle ;
	\draw   (180,512.06) .. controls (180,475.02) and (248.28,445) .. (332.5,445) .. controls (416.72,445) and (485,475.02) .. (485,512.06) .. controls (485,549.1) and (416.72,579.12) .. (332.5,579.12) .. controls (248.28,579.12) and (180,549.1) .. (180,512.06) -- cycle ;
	
	\draw (342,214) node [anchor=south east] [inner sep=0.75pt]   [align=left] {$\displaystyle v_{0}$};
	\draw (182,253) node [anchor=north west][inner sep=0.75pt]    {$U_{a}$};
	\draw (173,346.62) node [anchor=north west][inner sep=0.75pt]    {$v_{ai}^{'}$};
	\draw (327,382) node [anchor=north west][inner sep=0.75pt]    {$...$};
	\draw (433,367) node [anchor=north west][inner sep=0.75pt]    {$v_{bj}^{'}$};
	\draw (327,299) node [anchor=north west][inner sep=0.75pt]    {$...$};
	\draw (515,386) node [anchor=north west][inner sep=0.75pt]  [font=\scriptsize]  {$H[\hat{C}_{v_{bj}^{'}}]$};
	\draw (458.51,262.7) node [anchor=north west][inner sep=0.75pt]    {$U_{b}$};
	\draw (457,304) node [anchor=north west][inner sep=0.75pt]    {$j$};
	\draw (411,305) node [anchor=north west][inner sep=0.75pt]    {$1$};
	\draw (500,306.14) node [anchor=north west][inner sep=0.75pt]    {$n$};
	\draw (430.49,307.7) node [anchor=north west][inner sep=0.75pt]    {$...$};
	\draw (473.49,309.7) node [anchor=north west][inner sep=0.75pt]    {$...$};
	\draw (211,519) node [anchor=north east] [inner sep=0.75pt]    {$v_{ab}$};
	\draw (182,433) node [anchor=north west][inner sep=0.75pt]  [font=\scriptsize]  {$U_{a,b,i}$};
	\draw (457,518.38) node [anchor=north west][inner sep=0.75pt]    {$v_{ba}$};
	\draw (449,439) node [anchor=north west][inner sep=0.75pt]  [font=\scriptsize]  {$U_{b,a,j}$};
	\draw (106,368) node [anchor=north west][inner sep=0.75pt]  [font=\scriptsize]  {$H[\hat{C}_{v_{ai}^{'}}]$};
	\draw (517,495) node [anchor=north west][inner sep=0.75pt]  [rotate=-90]  {$...$};
	\draw (257,507) node [anchor=north west][inner sep=0.75pt]    {$U_{ij}^{ab}$};
	\draw (374,509) node [anchor=north west][inner sep=0.75pt]    {$U_{ji}^{ba}$};
	\draw (321,530) node [anchor=north west][inner sep=0.75pt]    {$v_{2}$};
	\draw (202.01,294.49) node [anchor=north west][inner sep=0.75pt]    {$i$};
	\draw (156.01,295.49) node [anchor=north west][inner sep=0.75pt]    {$1$};
	\draw (245.01,296.63) node [anchor=north west][inner sep=0.75pt]    {$n$};
	\draw (175.5,298.19) node [anchor=north west][inner sep=0.75pt]    {$...$};
	\draw (218.5,300.19) node [anchor=north west][inner sep=0.75pt]    {$...$};
	\draw (315,461.86) node [anchor=north west][inner sep=0.75pt]  [font=\scriptsize]  {$H[\hat{C}_{v_{2}}]$};
	\draw (343,163) node [anchor=north west][inner sep=0.75pt]    {$...$};
	\draw (145,504) node [anchor=north west][inner sep=0.75pt]  [rotate=-90]  {$...$};
	\draw (91,497) node [anchor=north west][inner sep=0.75pt]    {$U_{a,b}$};
	\draw (319,582) node [anchor=north west][inner sep=0.75pt]    {$H_{ij}^{ab}$};
	\draw (535,494) node [anchor=north west][inner sep=0.75pt]    {$U_{b,a}$};
	\draw (263,159) node [anchor=north west][inner sep=0.75pt]    {$U_{0}$};

\end{tikzpicture}
\end{center}
\caption{The graph $G'$. \label{fig:fvs}}
\end{figure}

Now, we prove that this is a parameterized reduction. 

\noindent\textbf{Proof of Sufficiency.}
Let $G=(U_1\cup\cdots\cup U_{\ell},E)$ be a yes-instance of \mcc\ and for all $a\in [\ell]$, $i_a$-th vertex in $U_a$ form a clique of size $\ell$ in $G$. We will construct a fair matching $M$ in $G'$ as follows. For each $a\in [\ell]$, $v_0$ is matched to $i_a$ in $U_a$ along with all vertices in $U_0$. Also, for each $a,b\in [\ell], a\neq b$, $v_{ab}$ is matched to all vertices in $U_{a,b,i_a}$ along with all vertices in $U^{ab}_{i_ai_b} \cup U_{a,b}$. Now, for each $a\in [\ell]$ and each $i\neq i_a$, the vertex $v'_{a,i}$ is matched to all its neighbors. Finally, for each $a,b\in [\ell], a\neq b$, each $i,j\in [n] $ where $i\neq i_a$, in the gadget $H^{ab}_{ij}$, the vertex $v_2$ is matched to all its neighbors. Every vertex in $U$ which is not yet matched, is in some $H[\hat{C}_v]$. So, we match these unmatched vertices to their pendent neighbors. It is clear that $M$ is a left-perfect matching. To see that $M$ is $L$-fair, note that if $v\in V\setminus \{v_0, v_{ab},a,b\in [\ell], a\neq b \}$, then $v$ is either matched to all its neighbors, or matched to none of its neighbors. Also, since the clique of $G$ has only one element in each $U_a$, $v_0$ is matched to one vertex of color $(a,0)$ and it pendent neighbors are of colors $(a,i)$, $i\in [n]$. So, $v_0$ is fair. Finally, the vertex $v_{ab}$ is matched to vertices in $U_{a,b,i_a}$ which are of colors $(b,1), \ldots, (b,i_a)$ along with vertices in $U^{ab}_{i_ai_b}$ which are of colors $(b,i_a+1),\ldots, (b,n)$. Other colors appear in vertices of $U_{a,b}$. So, $v_{ab} $ is also fair. 

\noindent\textbf{Proof of Necessity.}
Now, suppose that $G'$ has a left-perfect $L$-fair matching $M$. Since $M$ is left-perfect, all pendant vertices in $U$ are matched with their unique neighbors. Thus, $v_0$ has to be matched with only one vertex with color $(a,0)$ for each $a\in [\ell]$. Suppose that $v_0$ is matched to $i_a$-th vertex in $U_a$. We show that these vertices $i_1,\ldots, i_\ell$ form a clique in $G$. Let $a,b \in [\ell], a\neq b$. Since $i_a$ is matched with $v_0$, $v'_{a,i_a}$ is matched to no vertex. Therefore, $v_{ab}$ is matched to all vertices in $U_{a,b,i_a}$. Now, since $v_{ab}$ is fair, it should be matched with vertices of colors $(b,i_a+1),\ldots, (b,n)$. We show that these vertices should be in some gadget $H^{ab}_{i_aj}$, for some $j$. To see this, note that in the gadget $H^{ab}_{ij}$, the vertex $v_2$ is either matched to all its neighbors, or none of its neighbors. Therefore, $v_{ab}$ is either matched to all vertices in $U^{ab}_{ij}$, or none of vertices in $U^{ab}_{ij}$. This implies that $v_{ab}$ is matched to all vertices in $U^{ab}_{i_aj}$, for some $j$, in the gadget $H^{ab}_{i_aj}$. Now, we show that $j=i_b$. Note that the vertex $v_2$ in $H^{ab}_{i_aj}$ is matched to no vertex. So, vertices in $U^{ba}_{ji_a}$ are matched to $v_{ba}$. If $j\neq i_b$, then $v_{ba}$ is matched to either two neighbors of the same color $(a,i_b)$, or no vertex of color $(a,j)$. This contradiction implies that $j=i_b$ and therefore, $i_a$ and $i_b$ are adjacent. Therefore, $G=(U_1\cup\cdots\cup U_{\ell},E)$ is a yes-instance of \mcc.

Finally, note that in the graph $G'$, we have $\Delta_U=2$ and if one removes all vertices $v_{ab}$, $a,b\in [\ell]$ then the remaining graph is a forest consisting of trees of depth at most $4$. Thus, the feedback vertex number and the tree-depth of $G'$ is at most $O(\ell^2)$. Moreover, the number of vertices of $G'$ is at most $O(n\ell^3+n^2\ell^2+n\ell m)$, where $m$ is the number of edges of $G$. Hence, this is a parameterized reduction and the proof is complete. 
\end{proof}

\begin{theorem} \label{thm:pwCDelta}
\gfm\ is W[1]-hard with respect to $(\pw,|C|,\Delta_U)$ where $\pw$ is the path-width of the input graph,  $\Delta_U$ is the maximum degree of vertices in $U$ and $|C|$ is the number of colors. 
\end{theorem}

\begin{proof}
We give a reduction form \ubp\ which is known to be W[1]-hard with respect to the number of bins $m$ \cite{jansen}. 

\prob{Unary Bin Packing}{A number of positive integers $x_1,\ldots, x_n$, $m$ and $B$ which are given in unary encoding and we have $\sum_{i=1}^n{x_i}=mB$.}{Can we partition $x_1,\ldots, x_n$ into $m$ bins of capacity $B$?}

Let $(x_1,\ldots, x_n;m,B)$ be an instance of \ubp\ and we construct an instance of \gfm\ where the path-width of the graph as well as the number of colors and the maximum degree of vertices in $U$ is $O(m)$. 

First, let $C=[m+1]$ be the set of colors and $H=(U_H \cup V_H;E_H)$ be the bipartite graph constructed as follows. The vertex set $U_H$ consists of the disjoint union of sets $X$, $Y_i$ and $X_{ij}$, $i\in [n]$ and $j\in [m]$, where $|X|=B$, $|X_{ij}|=x_i$ and $|Y_i|=(m-1)x_i$ for each $i\in [n]$ and $j\in [m]$. Each vertex in   $X_{ij}$ has color $j$ and each vertex in $X\cup Y_i$ has color $m+1$. Also, we have $V_H=\{v_0, v_{ij}, i\in [n], j\in [m] \}$. For all vertices $v\in V_H$ we have $L(v)=0$. Now, the edges of $H$ are defined as follows. First, $v_0$ is complete to $X$ and $\cup_{i=1}^n \cup_{j=1}^m X_{ij}$. Also, each vertex $v_{ij}$ is complete to $X_{ij}\cup Y_i$ (see Figure~\ref{fig:pw}). We also construct two gadgets as follows.

\paragraph*{Gadget $H[\hat{C},k]$}
For every positive integer $k$ and every subset of colors $\hat{C}\subseteq C$, we define a gadget $H[\hat{C},k]$ as follows. The bipartite graph $H[\hat{C},k]$ has a bipartition $(U,V)$ where $|U|=|\hat{C}|k$ and for each color $c\in \hat{C}$, there are $k$ vertices $u^1_c,\ldots, u^k_c\in U$ of color $c$. There is a vertex $\bar{v}\in V$ which is complete to $U$. Also, each vertex $u^i_c\in U$ has a private neighbor $v^i_c$ in $V$. Also, $L(v)=0$ and $L(v^i_c)=1$ for all $c\in \hat{C}$. There is no more edges in $H[\hat{C},k]$ (see Figure~\ref{fig:pwgad1}).  
       
\begin{figure}[ht]
\begin{center}
\begin{tikzpicture}[x=0.75pt,y=0.75pt,yscale=-1,xscale=1]
	
	\draw  [color={rgb, 255:red, 29; green, 2; blue, 2 }  ,draw opacity=1 ][fill={rgb, 255:red, 36; green, 4; blue, 4 }  ,fill opacity=1 ] (305.29,101.56) .. controls (305.27,99.35) and (303.19,97.58) .. (300.66,97.6) .. controls (298.12,97.62) and (296.08,99.42) .. (296.1,101.63) .. controls (296.13,103.84) and (298.2,105.61) .. (300.74,105.59) .. controls (303.28,105.57) and (305.32,103.77) .. (305.29,101.56) -- cycle ;
	\draw  [color={rgb, 255:red, 29; green, 2; blue, 2 }  ,draw opacity=1 ][fill={rgb, 255:red, 36; green, 4; blue, 4 }  ,fill opacity=1 ] (254.8,200.03) .. controls (254.78,197.82) and (252.7,196.05) .. (250.16,196.07) .. controls (247.63,196.09) and (245.59,197.89) .. (245.61,200.1) .. controls (245.64,202.31) and (247.71,204.09) .. (250.25,204.07) .. controls (252.79,204.05) and (254.83,202.24) .. (254.8,200.03) -- cycle ;
	\draw  [color={rgb, 255:red, 29; green, 2; blue, 2 }  ,draw opacity=1 ][fill={rgb, 255:red, 36; green, 4; blue, 4 }  ,fill opacity=1 ] (356.22,151.12) .. controls (356.19,148.92) and (354.11,147.14) .. (351.58,147.16) .. controls (349.04,147.18) and (347,148.99) .. (347.02,151.2) .. controls (347.05,153.41) and (349.12,155.18) .. (351.66,155.16) .. controls (354.2,155.14) and (356.24,153.33) .. (356.22,151.12) -- cycle ;
	\draw  [color={rgb, 255:red, 29; green, 2; blue, 2 }  ,draw opacity=1 ][fill={rgb, 255:red, 36; green, 4; blue, 4 }  ,fill opacity=1 ] (305.17,200.27) .. controls (305.15,198.06) and (303.07,196.29) .. (300.53,196.31) .. controls (297.99,196.33) and (295.96,198.13) .. (295.98,200.34) .. controls (296,202.55) and (298.08,204.32) .. (300.62,204.3) .. controls (303.16,204.28) and (305.2,202.48) .. (305.17,200.27) -- cycle ;
	\draw  [color={rgb, 255:red, 29; green, 2; blue, 2 }  ,draw opacity=1 ][fill={rgb, 255:red, 36; green, 4; blue, 4 }  ,fill opacity=1 ] (305.6,150.27) .. controls (305.58,148.06) and (303.5,146.29) .. (300.97,146.31) .. controls (298.43,146.33) and (296.39,148.13) .. (296.41,150.34) .. controls (296.44,152.55) and (298.51,154.33) .. (301.05,154.31) .. controls (303.59,154.29) and (305.63,152.48) .. (305.6,150.27) -- cycle ;
	\draw  [color={rgb, 255:red, 29; green, 2; blue, 2 }  ,draw opacity=1 ][fill={rgb, 255:red, 36; green, 4; blue, 4 }  ,fill opacity=1 ] (255.19,150.03) .. controls (255.17,147.82) and (253.09,146.05) .. (250.55,146.07) .. controls (248.01,146.09) and (245.98,147.9) .. (246,150.1) .. controls (246.02,152.31) and (248.1,154.09) .. (250.64,154.07) .. controls (253.18,154.05) and (255.22,152.24) .. (255.19,150.03) -- cycle ;
	\draw    (300.7,101.6) -- (250.6,150.07) ;
	\draw    (250.6,150.07) -- (250.21,200.07) ;
	\draw    (300.7,101.6) -- (351.62,151.16) ;
	\draw  [color={rgb, 255:red, 29; green, 2; blue, 2 }  ,draw opacity=1 ][fill={rgb, 255:red, 36; green, 4; blue, 4 }  ,fill opacity=1 ] (354.84,200.12) .. controls (354.81,197.91) and (352.74,196.14) .. (350.2,196.16) .. controls (347.66,196.18) and (345.62,197.99) .. (345.65,200.19) .. controls (345.67,202.4) and (347.75,204.18) .. (350.28,204.16) .. controls (352.82,204.14) and (354.86,202.33) .. (354.84,200.12) -- cycle ;
	\draw    (351.62,151.16) -- (350.24,200.16) ;
	\draw    (301.01,150.31) -- (300.58,200.3) ;
	\draw    (300.7,101.6) -- (301.01,150.31) ;
	
	\draw (306.33,86.68) node [anchor=north west][inner sep=0.75pt]  [rotate=-359.47]  {$\overline{v}$};
	\draw (250.64,154.07) node [anchor=north west][inner sep=0.75pt]  [rotate=-359.47]  {$u_{c}^{1}$};
	\draw (301.05,154.31) node [anchor=north west][inner sep=0.75pt]  [rotate=-359.47]  {$u_{c}^{2}$};
	\draw (351.66,155.16) node [anchor=north west][inner sep=0.75pt]  [rotate=-359.47]  {$u_{c}^{k}$};
	\draw (317.71,149.65) node [anchor=north west][inner sep=0.75pt]    {$...$};
	\draw (250.25,204.07) node [anchor=north west][inner sep=0.75pt]  [rotate=-359.47]  {$v_{c}^{1}$};
	\draw (300.62,204.3) node [anchor=north west][inner sep=0.75pt]  [rotate=-359.47]  {$v_{c}^{2}$};
	\draw (350.28,204.16) node [anchor=north west][inner sep=0.75pt]  [rotate=-359.47]  {$v_{c}^{k}$};
	\draw (317.71,198.67) node [anchor=north west][inner sep=0.75pt]    {$...$};
	
\end{tikzpicture}
\caption{The gadget graph $H[\hat{C},k]$. \label{fig:pwgad1}}	
\end{center}
\end{figure}
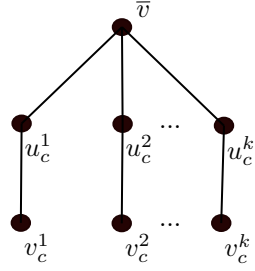

\paragraph*{Gadget $H_1[k,l]$}
Also, for each positive integer $k$ and each $l\in [m]$, we construct a gadget $H_1[k,l]$ as follows. The bipartite graph $H_1[k,l]$ has a bipartition $(U,V)$ where $U=\{u_i,u'_i,u''_j,u'''_j; i\in [k], j\in [k-2]\}$ and $V=\{\hat{v},v'_i, v''_j; i\in [2,k], j\in [2,k-1]\}$. For all $v\in V$, we have $L(v)=0$. All vertices $u_1,\ldots, u_k$ are of color $l$. The vertex $u'_1$ is of color $l+2 \pmod m$. For each $i\in [2, k]$, the vertex $u'_i$ is of color $l+1 \pmod m$. For each $j\in [k-2]$, the vertex $u''_j$ is of color $l$ and the vertex $u'''_j$ is of color $l+2$. The edges are as follows. The vertex $\hat{v}$ is adjacent to all $u_i$,  $i\in [k]$. Each vertex $v_i$ is adjacent to $u_i$ and $u'_i$. For each $i\in [2,k]$, $v'_i$ is adjacent to $u'_i$. Also, $v'_2$ is adjacent to $u'_1$. For each $j\in [k-2]$, $u''_j$ is adjacent to $v'_{j+1}$ and $v''_{j+1}$.  For each $j\in [k-2]$, $u'''_j$ is adjacent to $v''_{j+1}$ and $v'_{j+2}$. Now, for each vertex $v\in V$, let $\hat{C}_v$ be the set of colors in $C$ which are absent in neighbors of $v$.  For each vertex $v\in V\setminus \{\hat{v}\}$, we add a gadget $H[\hat{C}_v,1]$ and identify $\bar{v}$ and $v$ (see Figure~\ref{fig:pwgad2}).

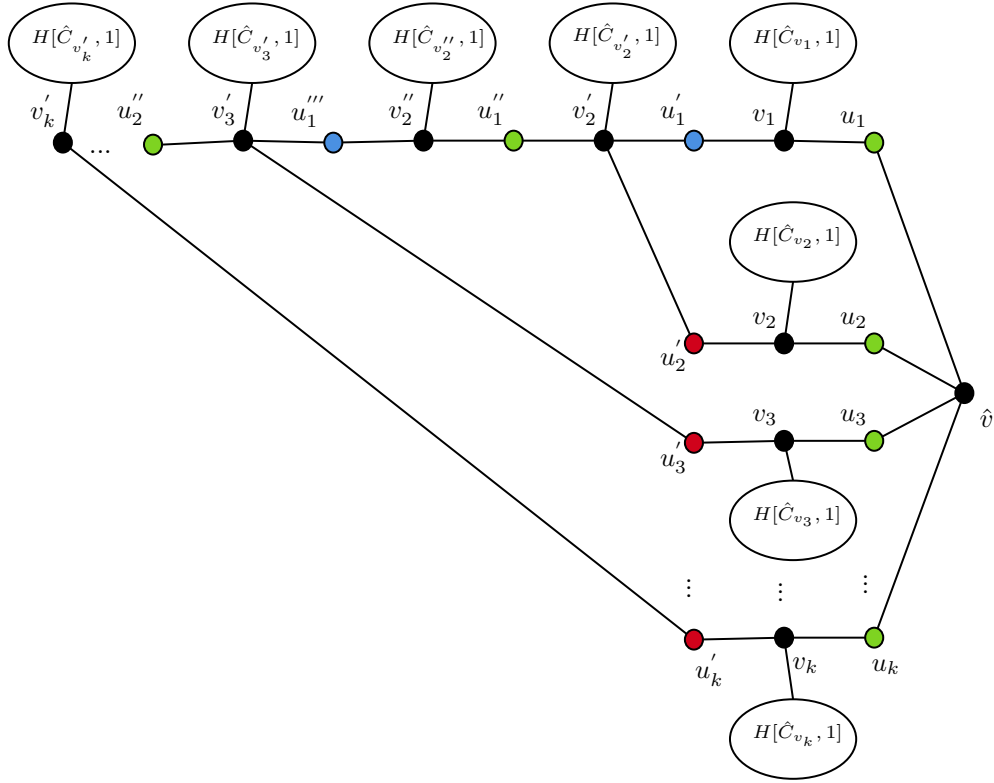
\begin{figure}[ht]
\begin{center}   
\begin{tikzpicture}[x=0.75pt,y=0.75pt,yscale=-1,xscale=.9]
	
	\draw    (500,200) -- (450,199) ;
	\draw    (500,301) -- (450,301) ;
	\draw    (500,449) -- (450,449) ;
	\draw    (450,199) -- (400,199) ;
	\draw    (450,301) -- (400,301) ;
	\draw    (450,449) -- (400,450) ;
	\draw    (550,326) -- (500,301) ;
	\draw    (550,326) -- (500,200) ;
	\draw    (500,449) -- (550,326) ;
	\draw    (350,199) -- (400,301) ;
	\draw    (400,199) -- (350,199) ;
	\draw    (350,199) -- (300,199) ;
	\draw    (300,199) -- (250,199) ;
	\draw    (500,350) -- (450,350) ;
	\draw    (450,350) -- (400,351) ;
	\draw    (500,350) -- (550,326) ;
	\draw    (50,200) -- (400,450) ;
	\draw    (400,351) -- (150,199) ;
	\draw    (250,199) -- (200,200) ;
	\draw    (200,200) -- (150,199) ;
	\draw  [fill={rgb, 255:red, 0; green, 0; blue, 0 }  ,fill opacity=1 ] (545,326) .. controls (545,323.24) and (547.24,321) .. (550,321) .. controls (552.76,321) and (555,323.24) .. (555,326) .. controls (555,328.76) and (552.76,331) .. (550,331) .. controls (547.24,331) and (545,328.76) .. (545,326) -- cycle ;
	\draw  [fill={rgb, 255:red, 126; green, 211; blue, 33 }  ,fill opacity=1 ] (495,200) .. controls (495,197.24) and (497.24,195) .. (500,195) .. controls (502.76,195) and (505,197.24) .. (505,200) .. controls (505,202.76) and (502.76,205) .. (500,205) .. controls (497.24,205) and (495,202.76) .. (495,200) -- cycle ;
	\draw  [fill={rgb, 255:red, 0; green, 0; blue, 0 }  ,fill opacity=1 ] (445,199) .. controls (445,196.24) and (447.24,194) .. (450,194) .. controls (452.76,194) and (455,196.24) .. (455,199) .. controls (455,201.76) and (452.76,204) .. (450,204) .. controls (447.24,204) and (445,201.76) .. (445,199) -- cycle ;
	\draw  [fill={rgb, 255:red, 74; green, 144; blue, 226 }  ,fill opacity=1 ] (395,199) .. controls (395,196.24) and (397.24,194) .. (400,194) .. controls (402.76,194) and (405,196.24) .. (405,199) .. controls (405,201.76) and (402.76,204) .. (400,204) .. controls (397.24,204) and (395,201.76) .. (395,199) -- cycle ;
	\draw  [fill={rgb, 255:red, 0; green, 0; blue, 0 }  ,fill opacity=1 ] (345,199) .. controls (345,196.24) and (347.24,194) .. (350,194) .. controls (352.76,194) and (355,196.24) .. (355,199) .. controls (355,201.76) and (352.76,204) .. (350,204) .. controls (347.24,204) and (345,201.76) .. (345,199) -- cycle ;
	\draw  [fill={rgb, 255:red, 126; green, 211; blue, 33 }  ,fill opacity=1 ] (295,199) .. controls (295,196.24) and (297.24,194) .. (300,194) .. controls (302.76,194) and (305,196.24) .. (305,199) .. controls (305,201.76) and (302.76,204) .. (300,204) .. controls (297.24,204) and (295,201.76) .. (295,199) -- cycle ;
	\draw  [fill={rgb, 255:red, 0; green, 0; blue, 0 }  ,fill opacity=1 ] (245,199) .. controls (245,196.24) and (247.24,194) .. (250,194) .. controls (252.76,194) and (255,196.24) .. (255,199) .. controls (255,201.76) and (252.76,204) .. (250,204) .. controls (247.24,204) and (245,201.76) .. (245,199) -- cycle ;
	\draw  [fill={rgb, 255:red, 74; green, 144; blue, 226 }  ,fill opacity=1 ] (195,200) .. controls (195,197.24) and (197.24,195) .. (200,195) .. controls (202.76,195) and (205,197.24) .. (205,200) .. controls (205,202.76) and (202.76,205) .. (200,205) .. controls (197.24,205) and (195,202.76) .. (195,200) -- cycle ;
	\draw  [fill={rgb, 255:red, 0; green, 0; blue, 0 }  ,fill opacity=1 ] (145,199) .. controls (145,196.24) and (147.24,194) .. (150,194) .. controls (152.76,194) and (155,196.24) .. (155,199) .. controls (155,201.76) and (152.76,204) .. (150,204) .. controls (147.24,204) and (145,201.76) .. (145,199) -- cycle ;
	\draw  [fill={rgb, 255:red, 126; green, 211; blue, 33 }  ,fill opacity=1 ] (495,301) .. controls (495,298.24) and (497.24,296) .. (500,296) .. controls (502.76,296) and (505,298.24) .. (505,301) .. controls (505,303.76) and (502.76,306) .. (500,306) .. controls (497.24,306) and (495,303.76) .. (495,301) -- cycle ;
	\draw  [fill={rgb, 255:red, 0; green, 0; blue, 0 }  ,fill opacity=1 ] (445,301) .. controls (445,298.24) and (447.24,296) .. (450,296) .. controls (452.76,296) and (455,298.24) .. (455,301) .. controls (455,303.76) and (452.76,306) .. (450,306) .. controls (447.24,306) and (445,303.76) .. (445,301) -- cycle ;
	\draw  [fill={rgb, 255:red, 208; green, 2; blue, 27 }  ,fill opacity=1 ] (395,301) .. controls (395,298.24) and (397.24,296) .. (400,296) .. controls (402.76,296) and (405,298.24) .. (405,301) .. controls (405,303.76) and (402.76,306) .. (400,306) .. controls (397.24,306) and (395,303.76) .. (395,301) -- cycle ;
	\draw  [fill={rgb, 255:red, 126; green, 211; blue, 33 }  ,fill opacity=1 ] (495,350) .. controls (495,347.24) and (497.24,345) .. (500,345) .. controls (502.76,345) and (505,347.24) .. (505,350) .. controls (505,352.76) and (502.76,355) .. (500,355) .. controls (497.24,355) and (495,352.76) .. (495,350) -- cycle ;
	\draw  [fill={rgb, 255:red, 208; green, 2; blue, 27 }  ,fill opacity=1 ] (395,351) .. controls (395,348.24) and (397.24,346) .. (400,346) .. controls (402.76,346) and (405,348.24) .. (405,351) .. controls (405,353.76) and (402.76,356) .. (400,356) .. controls (397.24,356) and (395,353.76) .. (395,351) -- cycle ;
	\draw  [fill={rgb, 255:red, 0; green, 0; blue, 0 }  ,fill opacity=1 ] (445,350) .. controls (445,347.24) and (447.24,345) .. (450,345) .. controls (452.76,345) and (455,347.24) .. (455,350) .. controls (455,352.76) and (452.76,355) .. (450,355) .. controls (447.24,355) and (445,352.76) .. (445,350) -- cycle ;
	\draw  [fill={rgb, 255:red, 208; green, 2; blue, 27 }  ,fill opacity=1 ] (395,450) .. controls (395,447.24) and (397.24,445) .. (400,445) .. controls (402.76,445) and (405,447.24) .. (405,450) .. controls (405,452.76) and (402.76,455) .. (400,455) .. controls (397.24,455) and (395,452.76) .. (395,450) -- cycle ;
	\draw  [fill={rgb, 255:red, 0; green, 0; blue, 0 }  ,fill opacity=1 ] (445,449) .. controls (445,446.24) and (447.24,444) .. (450,444) .. controls (452.76,444) and (455,446.24) .. (455,449) .. controls (455,451.76) and (452.76,454) .. (450,454) .. controls (447.24,454) and (445,451.76) .. (445,449) -- cycle ;
	\draw  [fill={rgb, 255:red, 126; green, 211; blue, 33 }  ,fill opacity=1 ] (495,449) .. controls (495,446.24) and (497.24,444) .. (500,444) .. controls (502.76,444) and (505,446.24) .. (505,449) .. controls (505,451.76) and (502.76,454) .. (500,454) .. controls (497.24,454) and (495,451.76) .. (495,449) -- cycle ;
	\draw  [fill={rgb, 255:red, 126; green, 211; blue, 33 }  ,fill opacity=1 ] (95,201) .. controls (95,198.24) and (97.24,196) .. (100,196) .. controls (102.76,196) and (105,198.24) .. (105,201) .. controls (105,203.76) and (102.76,206) .. (100,206) .. controls (97.24,206) and (95,203.76) .. (95,201) -- cycle ;
	\draw  [fill={rgb, 255:red, 0; green, 0; blue, 0 }  ,fill opacity=1 ] (45,200) .. controls (45,197.24) and (47.24,195) .. (50,195) .. controls (52.76,195) and (55,197.24) .. (55,200) .. controls (55,202.76) and (52.76,205) .. (50,205) .. controls (47.24,205) and (45,202.76) .. (45,200) -- cycle ;
	\draw    (105,201) -- (150,199) ;
	\draw  [color={rgb, 255:red, 0; green, 0; blue, 0 }  ,draw opacity=1 ] (420,150) .. controls (420,138.95) and (435.67,130) .. (455,130) .. controls (474.33,130) and (490,138.95) .. (490,150) .. controls (490,161.05) and (474.33,170) .. (455,170) .. controls (435.67,170) and (420,161.05) .. (420,150) -- cycle ;
	\draw  [color={rgb, 255:red, 0; green, 0; blue, 0 }  ,draw opacity=1 ] (420,250) .. controls (420,238.95) and (435.67,230) .. (455,230) .. controls (474.33,230) and (490,238.95) .. (490,250) .. controls (490,261.05) and (474.33,270) .. (455,270) .. controls (435.67,270) and (420,261.05) .. (420,250) -- cycle ;
	\draw  [color={rgb, 255:red, 0; green, 0; blue, 0 }  ,draw opacity=1 ] (420,390) .. controls (420,378.95) and (435.67,370) .. (455,370) .. controls (474.33,370) and (490,378.95) .. (490,390) .. controls (490,401.05) and (474.33,410) .. (455,410) .. controls (435.67,410) and (420,401.05) .. (420,390) -- cycle ;
	\draw  [color={rgb, 255:red, 0; green, 0; blue, 0 }  ,draw opacity=1 ] (420,500) .. controls (420,488.95) and (435.67,480) .. (455,480) .. controls (474.33,480) and (490,488.95) .. (490,500) .. controls (490,511.05) and (474.33,520) .. (455,520) .. controls (435.67,520) and (420,511.05) .. (420,500) -- cycle ;
	\draw    (455,170) -- (450,199) ;
	\draw    (455,270) -- (450,301) ;
	\draw    (450,350) -- (455,370) ;
	\draw    (450,449) -- (455,480) ;
	\draw  [color={rgb, 255:red, 0; green, 0; blue, 0 }  ,draw opacity=1 ] (320,150) .. controls (320,138.95) and (335.67,130) .. (355,130) .. controls (374.33,130) and (390,138.95) .. (390,150) .. controls (390,161.05) and (374.33,170) .. (355,170) .. controls (335.67,170) and (320,161.05) .. (320,150) -- cycle ;
	\draw    (355,170) -- (350,199) ;
	\draw  [color={rgb, 255:red, 0; green, 0; blue, 0 }  ,draw opacity=1 ] (220,150) .. controls (220,138.95) and (235.67,130) .. (255,130) .. controls (274.33,130) and (290,138.95) .. (290,150) .. controls (290,161.05) and (274.33,170) .. (255,170) .. controls (235.67,170) and (220,161.05) .. (220,150) -- cycle ;
	\draw    (255,170) -- (250,199) ;
	\draw  [color={rgb, 255:red, 0; green, 0; blue, 0 }  ,draw opacity=1 ] (120,150) .. controls (120,138.95) and (135.67,130) .. (155,130) .. controls (174.33,130) and (190,138.95) .. (190,150) .. controls (190,161.05) and (174.33,170) .. (155,170) .. controls (135.67,170) and (120,161.05) .. (120,150) -- cycle ;
	\draw    (155,170) -- (150,199) ;
	\draw  [color={rgb, 255:red, 0; green, 0; blue, 0 }  ,draw opacity=1 ] (20,150) .. controls (20,138.95) and (35.67,130) .. (55,130) .. controls (74.33,130) and (90,138.95) .. (90,150) .. controls (90,161.05) and (74.33,170) .. (55,170) .. controls (35.67,170) and (20,161.05) .. (20,150) -- cycle ;
	\draw    (55,170) -- (50,200) ;
	
	\draw (556.84,332.39) node [anchor=north west][inner sep=0.75pt]  [rotate=-359.47]  {$\hat{v}$};
	\draw (498,193.6) node [anchor=south east] [inner sep=0.75pt]    {$u_{1}$};
	\draw (498,343.6) node [anchor=south east] [inner sep=0.75pt]    {$u_{3}$};
	\draw (497.03,458.38) node [anchor=north west][inner sep=0.75pt]  [rotate=-359.47]  {$u_{k}$};
	\draw (447.87,192.7) node [anchor=south east] [inner sep=0.75pt]    {$v_{1}$};
	\draw (448,294.6) node [anchor=south east] [inner sep=0.75pt]    {$v_{2}$};
	\draw (453.03,457.38) node [anchor=north west][inner sep=0.75pt]  [rotate=-359.47]  {$v_{k}$};
	\draw (398,192.6) node [anchor=south east] [inner sep=0.75pt]    {$u_{1}^{'}$};
	\draw (398,307) node [anchor=east] [inner sep=0.75pt]    {$u_{2}^{'}$};
	\draw (399,455.4) node [anchor=north west][inner sep=0.75pt]    {$u_{k}^{'}$};
	\draw (348,192.6) node [anchor=south east] [inner sep=0.75pt]    {$v_{2}^{'}$};
	\draw (298,192.6) node [anchor=south east] [inner sep=0.75pt]    {$u_{1}^{''}$};
	\draw (248,192.6) node [anchor=south east] [inner sep=0.75pt]    {$v_{2}^{''}$};
	\draw (497.87,294.71) node [anchor=south east] [inner sep=0.75pt]    {$u_{2}$};
	\draw (448,343.6) node [anchor=south east] [inner sep=0.75pt]    {$v_{3}$};
	\draw (398,367.6) node [anchor=south east] [inner sep=0.75pt]    {$u_{3}^{'}$};
	\draw (198,193.6) node [anchor=south east] [inner sep=0.75pt]    {$u_{1}^{'''}$};
	\draw (148,192.6) node [anchor=south east] [inner sep=0.75pt]    {$v_{3}^{'}$};
	\draw (48,193.6) node [anchor=south east] [inner sep=0.75pt]    {$v_{k}^{'}$};
	\draw (98,192.6) node [anchor=south east] [inner sep=0.75pt]    {$u_{2}^{''}$};
	\draw (63,202.4) node [anchor=north west][inner sep=0.75pt]    {$...$};
	\draw (497.6,415) node [anchor=north west][inner sep=0.75pt]  [rotate=-90]  {$...$};
	\draw (450.6,419) node [anchor=north west][inner sep=0.75pt]  [rotate=-90]  {$...$};
	\draw (399.6,417) node [anchor=north west][inner sep=0.75pt]  [rotate=-90]  {$...$};
	\draw (431,139.4) node [anchor=north west][inner sep=0.75pt]  [font=\scriptsize]  {$H[\hat{C}_{v_{1}} ,1]$};
	\draw (431,239.4) node [anchor=north west][inner sep=0.75pt]  [font=\scriptsize]  {$H[\hat{C}_{v_{2}} ,1]$};
	\draw (431,379.4) node [anchor=north west][inner sep=0.75pt]  [font=\scriptsize]  {$H[\hat{C}_{v_{3}} ,1]$};
	\draw (431,489.4) node [anchor=north west][inner sep=0.75pt]  [font=\scriptsize]  {$H[\hat{C}_{v_{k}} ,1]$};
	\draw (331,138.4) node [anchor=north west][inner sep=0.75pt]  [font=\scriptsize]  {$H[\hat{C}_{v_{2}^{'}} ,1]$};
	\draw (231,139.4) node [anchor=north west][inner sep=0.75pt]  [font=\scriptsize]  {$H[\hat{C}_{v_{2}^{''}} ,1]$};
	\draw (131,139.4) node [anchor=north west][inner sep=0.75pt]  [font=\scriptsize]  {$H[\hat{C}_{v_{3}^{'}} ,1]$};
	\draw (31,139.4) node [anchor=north west][inner sep=0.75pt]  [font=\scriptsize]  {$H[\hat{C}_{v_{k}^{'}} ,1]$};
\end{tikzpicture}
\caption{The gadget graph $H_1[k,l]$. \label{fig:pwgad2}}	
\end{center}
\end{figure}

 Now, we are ready to construct the reduction. The graph $G$ is defined as follows. We take the graph $H$ and for each $i\in [n]$ and $j\in [m]$, we add a disjoint copy of the gadget $H_1[x_i,j+1]$ ($j+1$ is taken module $m$) and identify $\hat{v}$ to $v_{ij}$. Also, let $\hat{C}_{ij}$ be the set of colors in $C$ which are absent in neighbors of $v_{ij}$.  Finally, we add a gadget $H[\hat{C}_{ij},x_i]$ and identify $\bar{v}$ and $v_{ij}$. Let $G$ be the obtained graph (see Figure~\ref{fig:pw}). \\
 
 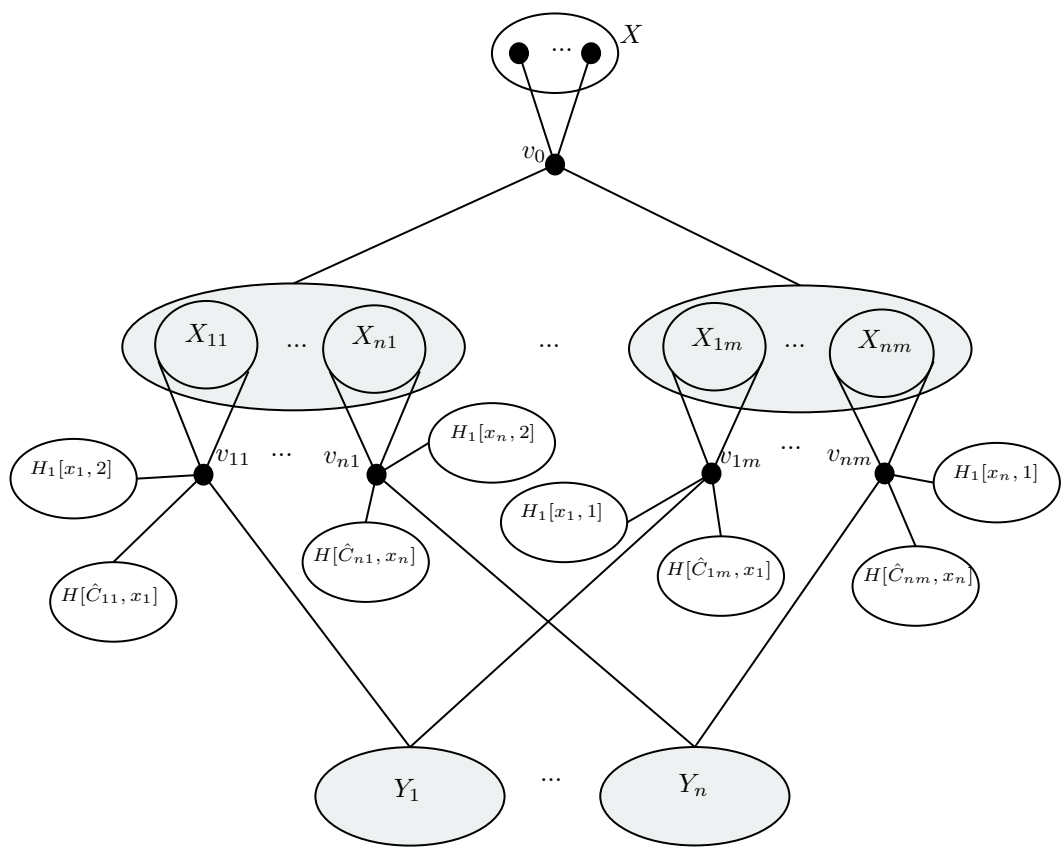
\begin{figure}[ht]
 	\begin{center}

\tikzset{every picture/.style={line width=0.75pt}} 

\begin{tikzpicture}[x=0.75pt,y=0.75pt,yscale=-1,xscale=.9]
	
	\draw  [color={rgb, 255:red, 0; green, 0; blue, 0 }  ,draw opacity=1 ][fill={rgb, 255:red, 238; green, 241; blue, 241 }  ,fill opacity=1 ] (110,302.72) .. controls (110,285.09) and (152.53,270.81) .. (204.99,270.81) .. controls (257.44,270.81) and (299.97,285.09) .. (299.97,302.72) .. controls (299.97,320.34) and (257.44,334.63) .. (204.99,334.63) .. controls (152.53,334.63) and (110,320.34) .. (110,302.72) -- cycle ;
	\draw   (128.23,301.62) .. controls (128.23,289.46) and (140.92,279.61) .. (156.58,279.61) .. controls (172.24,279.61) and (184.93,289.46) .. (184.93,301.62) .. controls (184.93,313.77) and (172.24,323.63) .. (156.58,323.63) .. controls (140.92,323.63) and (128.23,313.77) .. (128.23,301.62) -- cycle ;
	\draw   (221.39,303.82) .. controls (221.39,291.66) and (234.08,281.81) .. (249.74,281.81) .. controls (265.4,281.81) and (278.1,291.66) .. (278.1,303.82) .. controls (278.1,315.97) and (265.4,325.83) .. (249.74,325.83) .. controls (234.08,325.83) and (221.39,315.97) .. (221.39,303.82) -- cycle ;
	\draw  [color={rgb, 255:red, 0; green, 0; blue, 0 }  ,draw opacity=1 ][fill={rgb, 255:red, 238; green, 241; blue, 241 }  ,fill opacity=1 ] (391.03,303.72) .. controls (391.03,286.1) and (433.56,271.81) .. (486.01,271.81) .. controls (538.47,271.81) and (581,286.1) .. (581,303.72) .. controls (581,321.35) and (538.47,335.63) .. (486.01,335.63) .. controls (433.56,335.63) and (391.03,321.35) .. (391.03,303.72) -- cycle ;
	\draw   (410.07,302.62) .. controls (410.07,290.47) and (422.76,280.61) .. (438.42,280.61) .. controls (454.08,280.61) and (466.77,290.47) .. (466.77,302.62) .. controls (466.77,314.78) and (454.08,324.63) .. (438.42,324.63) .. controls (422.76,324.63) and (410.07,314.78) .. (410.07,302.62) -- cycle ;
	\draw   (502.42,305.92) .. controls (502.42,293.77) and (515.31,283.91) .. (531.21,283.91) .. controls (547.11,283.91) and (560,293.77) .. (560,305.92) .. controls (560,318.08) and (547.11,327.93) .. (531.21,327.93) .. controls (515.31,327.93) and (502.42,318.08) .. (502.42,305.92) -- cycle ;
	\draw  [color={rgb, 255:red, 0; green, 0; blue, 0 }  ,draw opacity=1 ][fill={rgb, 255:red, 238; green, 241; blue, 241 }  ,fill opacity=1 ] (217.09,528.81) .. controls (217.09,515.14) and (240.58,504.05) .. (269.55,504.05) .. controls (298.52,504.05) and (322,515.14) .. (322,528.81) .. controls (322,542.49) and (298.52,553.57) .. (269.55,553.57) .. controls (240.58,553.57) and (217.09,542.49) .. (217.09,528.81) -- cycle ;
	\draw  [color={rgb, 255:red, 0; green, 0; blue, 0 }  ,draw opacity=1 ][fill={rgb, 255:red, 238; green, 241; blue, 241 }  ,fill opacity=1 ] (375.09,528.74) .. controls (375.09,515.06) and (398.58,503.98) .. (427.55,503.98) .. controls (456.52,503.98) and (480,515.06) .. (480,528.74) .. controls (480,542.41) and (456.52,553.49) .. (427.55,553.49) .. controls (398.58,553.49) and (375.09,542.41) .. (375.09,528.74) -- cycle ;
	\draw  [fill={rgb, 255:red, 0; green, 0; blue, 0 }  ,fill opacity=1 ] (345,211) .. controls (345,208.24) and (347.24,206) .. (350,206) .. controls (352.76,206) and (355,208.24) .. (355,211) .. controls (355,213.76) and (352.76,216) .. (350,216) .. controls (347.24,216) and (345,213.76) .. (345,211) -- cycle ;
	\draw    (350,211) -- (204.99,270.81) ;
	\draw    (350,211) -- (486.01,271.81) ;
	\draw  [fill={rgb, 255:red, 0; green, 0; blue, 0 }  ,fill opacity=1 ] (246,367) .. controls (246,364.24) and (248.24,362) .. (251,362) .. controls (253.76,362) and (256,364.24) .. (256,367) .. controls (256,369.76) and (253.76,372) .. (251,372) .. controls (248.24,372) and (246,369.76) .. (246,367) -- cycle ;
	\draw  [fill={rgb, 255:red, 0; green, 0; blue, 0 }  ,fill opacity=1 ] (150,367) .. controls (150,364.24) and (152.24,362) .. (155,362) .. controls (157.76,362) and (160,364.24) .. (160,367) .. controls (160,369.76) and (157.76,372) .. (155,372) .. controls (152.24,372) and (150,369.76) .. (150,367) -- cycle ;
	\draw    (130,310) -- (155,367) ;
	\draw    (180,315) -- (155,367) ;
	\draw    (225,315) -- (251,367) ;
	\draw    (275,315) -- (251,367) ;
	\draw  [color={rgb, 255:red, 0; green, 0; blue, 0 }  ,draw opacity=1 ] (48,369) .. controls (48,357.95) and (63.67,349) .. (83,349) .. controls (102.33,349) and (118,357.95) .. (118,369) .. controls (118,380.05) and (102.33,389) .. (83,389) .. controls (63.67,389) and (48,380.05) .. (48,369) -- cycle ;
	\draw  [fill={rgb, 255:red, 0; green, 0; blue, 0 }  ,fill opacity=1 ] (527.77,366.38) .. controls (527.77,363.62) and (530.01,361.38) .. (532.77,361.38) .. controls (535.53,361.38) and (537.77,363.62) .. (537.77,366.38) .. controls (537.77,369.14) and (535.53,371.38) .. (532.77,371.38) .. controls (530.01,371.38) and (527.77,369.14) .. (527.77,366.38) -- cycle ;
	\draw  [fill={rgb, 255:red, 0; green, 0; blue, 0 }  ,fill opacity=1 ] (431.77,366.38) .. controls (431.77,363.62) and (434.01,361.38) .. (436.77,361.38) .. controls (439.53,361.38) and (441.77,363.62) .. (441.77,366.38) .. controls (441.77,369.14) and (439.53,371.38) .. (436.77,371.38) .. controls (434.01,371.38) and (431.77,369.14) .. (431.77,366.38) -- cycle ;
	\draw    (415,315) -- (436.77,366.38) ;
	\draw    (465,310) -- (436.77,366.38) ;
	\draw    (505,315) -- (532.77,366.38) ;
	\draw    (560,310) -- (532.77,366.38) ;
	\draw    (118,369) -- (155,367) ;
	\draw  [color={rgb, 255:red, 0; green, 0; blue, 0 }  ,draw opacity=1 ] (280,351) .. controls (280,339.95) and (295.67,331) .. (315,331) .. controls (334.33,331) and (350,339.95) .. (350,351) .. controls (350,362.05) and (334.33,371) .. (315,371) .. controls (295.67,371) and (280,362.05) .. (280,351) -- cycle ;
	\draw    (251,367) -- (280,351) ;
	\draw  [color={rgb, 255:red, 0; green, 0; blue, 0 }  ,draw opacity=1 ] (320,391) .. controls (320,379.95) and (335.67,371) .. (355,371) .. controls (374.33,371) and (390,379.95) .. (390,391) .. controls (390,402.05) and (374.33,411) .. (355,411) .. controls (335.67,411) and (320,402.05) .. (320,391) -- cycle ;
	\draw    (436.77,366.38) -- (269.55,504.05) ;
	\draw  [color={rgb, 255:red, 0; green, 0; blue, 0 }  ,draw opacity=1 ] (560,371) .. controls (560,359.95) and (575.67,351) .. (595,351) .. controls (614.33,351) and (630,359.95) .. (630,371) .. controls (630,382.05) and (614.33,391) .. (595,391) .. controls (575.67,391) and (560,382.05) .. (560,371) -- cycle ;
	\draw    (532.77,366.38) -- (560,371) ;
	\draw    (155,367) -- (269.55,504.05) ;
	\draw    (251,367) -- (427.55,503.98) ;
	\draw    (532.77,366.38) -- (427.55,503.98) ;
	\draw  [color={rgb, 255:red, 0; green, 0; blue, 0 }  ,draw opacity=1 ] (70,431) .. controls (70,419.95) and (85.67,411) .. (105,411) .. controls (124.33,411) and (140,419.95) .. (140,431) .. controls (140,442.05) and (124.33,451) .. (105,451) .. controls (85.67,451) and (70,442.05) .. (70,431) -- cycle ;
	\draw    (155,367) -- (105,411) ;
	\draw  [color={rgb, 255:red, 0; green, 0; blue, 0 }  ,draw opacity=1 ] (210,411) .. controls (210,399.95) and (225.67,391) .. (245,391) .. controls (264.33,391) and (280,399.95) .. (280,411) .. controls (280,422.05) and (264.33,431) .. (245,431) .. controls (225.67,431) and (210,422.05) .. (210,411) -- cycle ;
	\draw    (251,367) -- (245,391) ;
	\draw  [color={rgb, 255:red, 0; green, 0; blue, 0 }  ,draw opacity=1 ] (407,418) .. controls (407,406.95) and (422.67,398) .. (442,398) .. controls (461.33,398) and (477,406.95) .. (477,418) .. controls (477,429.05) and (461.33,438) .. (442,438) .. controls (422.67,438) and (407,429.05) .. (407,418) -- cycle ;
	\draw    (436.77,366.38) -- (442,398) ;
	\draw    (390,391) -- (436.77,366.38) ;
	\draw  [color={rgb, 255:red, 0; green, 0; blue, 0 }  ,draw opacity=1 ] (515,423) .. controls (515,411.95) and (530.67,403) .. (550,403) .. controls (569.33,403) and (585,411.95) .. (585,423) .. controls (585,434.05) and (569.33,443) .. (550,443) .. controls (530.67,443) and (515,434.05) .. (515,423) -- cycle ;
	\draw    (532.77,366.38) -- (550,403) ;
	\draw    (330,155) -- (350,211) ;
	\draw  [fill={rgb, 255:red, 0; green, 0; blue, 0 }  ,fill opacity=1 ] (365,155) .. controls (365,152.24) and (367.24,150) .. (370,150) .. controls (372.76,150) and (375,152.24) .. (375,155) .. controls (375,157.76) and (372.76,160) .. (370,160) .. controls (367.24,160) and (365,157.76) .. (365,155) -- cycle ;
	\draw  [fill={rgb, 255:red, 0; green, 0; blue, 0 }  ,fill opacity=1 ] (325,155) .. controls (325,152.24) and (327.24,150) .. (330,150) .. controls (332.76,150) and (335,152.24) .. (335,155) .. controls (335,157.76) and (332.76,160) .. (330,160) .. controls (327.24,160) and (325,157.76) .. (325,155) -- cycle ;
	\draw  [color={rgb, 255:red, 0; green, 0; blue, 0 }  ,draw opacity=1 ] (315,155) .. controls (315,143.95) and (330.67,135) .. (350,135) .. controls (369.33,135) and (385,143.95) .. (385,155) .. controls (385,166.05) and (369.33,175) .. (350,175) .. controls (330.67,175) and (315,166.05) .. (315,155) -- cycle ;
	\draw    (370,155) -- (350,211) ;
	
	\draw (347,211) node [anchor=south east] [inner sep=0.75pt]   [align=left] {$\displaystyle v_{0}$};
	\draw (143.01,289.62) node [anchor=north west][inner sep=0.75pt]    {$X_{11}$};
	\draw (235.36,290.72) node [anchor=north west][inner sep=0.75pt]    {$X_{n1}$};
	\draw (423.76,291.72) node [anchor=north west][inner sep=0.75pt]    {$X_{1m}$};
	\draw (516.01,292.82) node [anchor=north west][inner sep=0.75pt]    {$X_{nm}$};
	\draw (160.1,352.95) node [anchor=north west][inner sep=0.75pt]    {$v_{11}$};
	\draw (439,354) node [anchor=north west][inner sep=0.75pt]    {$v_{1m}$};
	\draw (258.87,518.46) node [anchor=north west][inner sep=0.75pt]    {$Y_{1}$};
	\draw (416.87,516.19) node [anchor=north west][inner sep=0.75pt]    {$Y_{n}$};
	\draw (199,300) node [anchor=north west][inner sep=0.75pt]    {$...$};
	\draw (190.57,354.85) node [anchor=north west][inner sep=0.75pt]    {$...$};
	\draw (473,351) node [anchor=north west][inner sep=0.75pt]    {$...$};
	\draw (475.22,299.82) node [anchor=north west][inner sep=0.75pt]    {$...$};
	\draw (220,354) node [anchor=north west][inner sep=0.75pt]    {$v_{n1}$};
	\draw (499,353) node [anchor=north west][inner sep=0.75pt]    {$v_{nm}$};
	\draw (338.98,300) node [anchor=north west][inner sep=0.75pt]    {$...$};
	\draw (340,518) node [anchor=north west][inner sep=0.75pt]    {$...$};
	\draw (57,358) node [anchor=north west][inner sep=0.75pt]  [font=\scriptsize]  {$H_{1}[ x_{1} ,2]$};
	\draw (290,340) node [anchor=north west][inner sep=0.75pt]  [font=\scriptsize]  {$H_{1}[ x_{n} ,2]$};
	\draw (329,380) node [anchor=north west][inner sep=0.75pt]  [font=\scriptsize]  {$H_{1}[ x_{1} ,1]$};
	\draw (569,360) node [anchor=north west][inner sep=0.75pt]  [font=\scriptsize]  {$H_{1}[ x_{n} ,1]$};
	\draw (74,420) node [anchor=north west][inner sep=0.75pt]  [font=\scriptsize]  {$H[\hat{C}_{11} ,x_{1}]$};
	\draw (214,400) node [anchor=north west][inner sep=0.75pt]  [font=\scriptsize]  {$H[\hat{C}_{n1} ,x_{n}]$};
	\draw (410,407) node [anchor=north west][inner sep=0.75pt]  [font=\scriptsize]  {$H[\hat{C}_{1m} ,x_{1}]$};
	\draw (519,410) node [anchor=north west][inner sep=0.75pt]  [font=\scriptsize]  {$H[\hat{C}_{nm} ,x_{n}]$};
	\draw (384,139) node [anchor=north west][inner sep=0.75pt]    {$X$};
	\draw (345.98,151) node [anchor=north west][inner sep=0.75pt]    {$...$};

\end{tikzpicture}
 		\caption{The constructed graph $G$. \label{fig:pw}}	
 	\end{center}
 \end{figure}

\noindent\textbf{Proof of Sufficiency.}
Suppose that $(x_1,\ldots, x_n,m,B)$ is a yes-instance of \ubp\ and $A_1,\ldots, A_m$ are the bins such that $\sum_{x_i\in A_j} x_i=B$ for each $j\in [m]$. We are going to obtain a fair matching in $G$ as follows. Fix some $i\in [n]$. If $x_i$ is in bin $A_j$, then $v_0$ is matched to all vertices in $X_{ij}$. The vertex $v_0$ is also matched with $B$ neighbors of color $m+1$. Also, in the gadget $H_1[x_i,j+1]$ corresponding to $v_{ij}$, the vertices $v_1,\ldots, v_{x_i}$ and $v''_2,\ldots, v''_{x_i-1}$ are matched to all their neighbors.
For each $t \in [2,x_i]$ and each $c\in \hat{C}_{v'_t}$, in the gadget $H[\hat{C}_{v'_t},1]$ corresponding to $v'_{t}$, the vertex $v_c^1$ is matched to its unique neighbor.   
Finally, in the gadget $H[\hat{C}_{ij},x_i]$ corresponding to ${v_{ij}}$, the vertices $v_c^1,\ldots, v_c^{x_i}$ are matched to its unique neighbor, for all $c\in  \hat{C}_{ij}$. 

Now, for each $j'\in [m]\setminus \{j\}$, $v_{ij'}$ is matched to all vertices in $X_{ij'}$.  Also, $v_{ij'}$ is matched to $x_i$ vertices in $Y_i$ (this can be done since $|Y_i|=(m-1)x_i$). Now, in the gadget $H[\hat{C}_{ij'},x_i]$ corresponding to $v_{ij'}$, the vertex $\hat{v}$ (which is identified with $v_{ij'}$) is matched to all its neighbors. 
In the gadget $H_1[x_i,j'+1]$ corresponding to $v_{ij'}$, the vertices $v'_2,\ldots, v'_{x_i}$ and  $\hat{v}$ (which is identified with $v_{ij'}$) are matched to all its neighbors. Also, for each $v_t$ in $H_1[x_i,j'+1]$, $t\in [x_i]$, and for each $v''_t$ in $ H_1[x_i,j'+1]$, $t\in [2,x_i-1]$, in its corresponding gadget $H[\hat{C}_{v_t},1]$, the vertex $v_c^1$ is matched to its unique neighbor, for all $c\in \hat{C}_{v_t}$. 
Since the size of items in each bin is equal to $B$, $v_0$ has $B$ matched neighbor in each color. Other vertices $v\in V$ with $L(v)=0$ have equal matched neighbors in each color. Hence, the obtained matching is an $L$-fair matching.  \\

\noindent\textbf{Proof of necessity.} Suppose that $M$ is an $L$-fair matching for $G$. Fix some $i\in [n]$. For each $j\in [m]$, we have either $|M_{m+1}(v_{ij})|=x_i$ or $|M_{m+1}(v_{ij})|=0$. To see this, consider the vertex $v'_{x_i}$ in the gadget $H_1[x_i,j+1]$ corresponding to the vertex $v_{ij}$. Since $L(v'_{x_i})=0$, $v'_{x_i}$ is matched either with all its neighbors or none of them. In the former case, the vertices $v_1,\ldots v_{x_i}$ and $v''_2,\ldots v''_{x_i-1}$ are matched with no vertex and $v'_2,\ldots v'_{x_i-1}$ are matched to all their neighbors. Thus, $u_1,\ldots, u_{x_i}$ are matched with $v_{ij}$. Since $L(v_{ij})=0$, $v_{ij}$ has exactly $x_i$ matched neighbors in $Y_i$. In the latter case, since $M$ is left perfect, the vertices $v''_2,\ldots v''_{x_i-1}$ are matched to all its neighbors and $v'_2,\ldots, v'_{x_i}$ are matched to no vertex. So, $v_1,\ldots, v_{x_i}$ are matched to all its neighbors. Thus, $v_{ij}$ is not matched to $u_1,\ldots, u_{x_i}$ and so has no matched neighbor in $Y_i$. Hence, the number of matched neighbors of $v_{ij}$ in $Y_i$ is either $0$ or $x_i$. Thus, exactly $m-1$ vertices among $v_{i1},\ldots, v_{im}$ are matched with vertices in $Y_i$ and exactly one vertex say $v_{ij_0}$ is matched to no vertex in $Y_i$ and so is matched to no vertex in $X_{ij_0}$. Therefore, vertices in $X_{ij_0}$ are matched to $v_0$. Now, we put the item $x_i$ to the bin $A_{j_0}$. This gives a partition of the set $\{x_1,\ldots, x_n\}$ into $m$ bins $A_1,\ldots, A_m$. Now, since $v_0$ has $B$ matched neighbors in the color class $m+1$, $v_0$ has also $B$ matched neighbor in each color class $i\in [m]$. This implies that the sum of item sizes in each bin $A_i$ is equal to $B$ and thus $(x_1,\ldots, x_n;m,B)$ is a yes-instance of \ubp.   

Moreover, the path-width of the graph $H$ is at most $m+1$. To see this, note that if one remove $v_0$ from $H$ we obtain $n$ disjoint bipartite graphs where each one has a part of size $m$.
Also, each of the gadgets $H[\hat{C},k]$ and $H_1[k,l]$ has path-width at most $3$ (since $H[\hat{C},k]$ is a tree of depth two and if one remove $\hat{v}$ from $H_1[k,l]$, a subdivision of a caterpillar is obtained). Therefore, the path-width of the graph $G$ is at most $m+4$. Also, the number of colors is $m+1$ and all vertices in $U\cap V(G)$ are of degree at most $m$ (vertices in $Y_i$ have degree equal to $m$), so $\Delta_U=m$. Finally, the graph $H$ has at most $O(m^2B)$ vertices and each added gadget has at most $O(mB)$ vertices. So, the size of $G$ is at most a polynomial in $m$ and $B$. Also, the input of \ubp\ is given in unary encoding. Hence, this is a parameterized reduction and so \gfm\ is W[1]-hard with respect to the path-width of the input graph, the number of colors and $\Delta_U$. 
\end{proof}

\section{FPT Results}

In this section, we prove that \gfm\ is FPT with respect to treewidth and $\Delta_V$ (combined), tree-depth and number of colors (combined), $\fes$ and neighborhood diversity.

\begin{theorem} \label{thm:fes}
\textsc{Generalized Fair Matching} can be solved in time $O(2^{\fes}n\log n)$, where $n$ and $\fes$ are the number of vertices and the feedback edge number of the input graph, respectively.
\end{theorem}

\begin{proof}
Suppose that the input graph is a bipartite graph $G=(U\cup V,E)$ and $F$ is a subset of $E$ such that $|F|=\fes$ and $G-F$ is acyclic. Also, the coloring $\col: U\to C$ and a threshold function $L:V\to \mathbb{N}\cup \{0\}$ are given.
Fix a subset $F'\subseteq F$ such that $F'$ is a many-to-one matching and suppose that $F'$ is included in the solution. Also, define $U'$ as  a subset of $U$ consisting of vertices which are matched by edges in $F'$, i.e. $U'=\{u\in U: \exists uv\in F'\}$ and for every vertex $v\in V$ and $c\in C$, let $m_c^v$ be the number of vertices of color $c$ which are matched by $v$ in $F'$, i.e. $m_c^v=|\{u\in U_c: uv\in F' \}|$.
Now, remove from $G$ all vertices in $U'$ and all edges in $F$ and let $T$ be a connected component of the obtained graph which is a tree. For a subtree $T'$ of $T$ and a many-to-one matching $M$ for $T'$, we say that a vertex $v\in V$ is $\tilde{L}(v)$-fair in $M$, if we have $\max_{c,c'\in C} |M_c(v)|+m_c^v -|M_{c'}(v)|-m_{c'}^v \leq L(v)$. Also, we say that $M$ is $\tilde{L}$-fair if all of vertices $v\in V\cap V(T')$ are $\tilde{L}(v)$-fair in $M$.

Let $T$ be rooted at an arbitrary vertex $u_0\in U$. 
For every vertex $w\in U\cup V$ with parent $z\in U\cup V$, we define $p(w)=z$ (for convenience, we define $p(u_0)=u_0$) and $T_w$ as the subtree of $T$ consisting of all vertices in $T$ which are descendant of $w$ (including $w$ itself) and also let $T_{wz}$ be obtained from $T_w$ by adding the edge $wz$.

We define two functions $f:U\cup V\to \{0,1\}$ and $g:U\cup V\to \{0,1\}$ as follows. First, for every $w\in U\cup V$, we define $f(w)=1$ if and only if there exists an $\tilde{L}$-fair left-perfect many-to-one matching for $T_{w}$. Also, for every $u\in U$, we define $ g(u)=1$ if and only if there exists a left-perfect many-to-one matching $M$ for $T_{up(u)}$ such that $M(u)=p(u)$ and for every vertex $v\in V(T_u)\cap V$, $v$ is $\tilde{L}(v)$-fair in $M$. Moreover, for every $v\in V$, we define $g(v)=1$ if and only if there exists an $\tilde{L}$-fair left-perfect many-to-one matching $M$ for $T_{vp(v)}$ (note that in this case $p(v)\in M(v)$). It is clear that the answer to \gfm\ is yes if and only if there exists an $F'\subseteq F$ such that $f(u_0)=1$. 

Now, we compute the value of the functions $f,g$ recursively by traversing the tree upward. First, suppose that $u\in U$ has $k$ children $v_1,\ldots, v_k$. If there exists some $i\in [k]$ such that $g(v_i)=1$ and $f(v_j)=1$ for all $j\in [k]\setminus \{i\}$, then we have $f(u)=1$ (since $u$ has to matched with one of vertices $v_1,\ldots, v_k$), otherwise we define $f(u)=0$. Also, if for all $j\in [k]$ we have $f(v_j)=1$, then we have $g(u)=1$ (since $u$ is matched to its parent and no vertex $v_j$ is matched to $u$). Otherwise, we have $g(u)=0$. 

Now, let $v\in V$ have $k$ children $u_1,\ldots, u_k$. If there exists $i\in [k]$ such that $f(u_i)=g(u_i)=0$, then there is no feasible solution and the answer is no. So, suppose that for each $u_i$, either $f(u_i)=1$ or $g(u_i)=1$. Now, for each $c\in C$, let $n^1_c$ be the number of vertices $u_i$ of color $c$ such that $f(u_i)=0$ and $g(u_i)=1$ and let $n^2_c$ be the number of vertices $u_i$ of color $c$ such that $f(u_i)=1$ and $g(u_i)=1$. Note that if $f(u_i)=0$ and $g(u_i)=1$, then $u_i$ should be matched with its parent $v$ and if $f(u_i)=1$ and $g(u_i)=1$, then $u_i$ is allowed to be matched with its parent $v$. Therefore, in any $\tilde{L}$-fair matching $M$ for $T_v$, we have $|M_c(v)| \in [n^1_c,n^1_c+n^2_c]$. Define $A=\max_{c\in C} (n^1_c+m_c^v)$ and $B=\min_{c\in C} (n^1_c+n^2_c+m_c^v)$. We claim that $f(v)=1$ if and only if $A-B\leq L(v)$. To see this, note that if $A-B>L(v)$, then there are some colors $c,c'$ such that $|M_c(v)|+m_c^v-|M_{c'}(v)|-m_{c'}^v\geq n^1_c+m_c^v- (n^1_{c'}+n^2_{c'}+m_{c'}^v)=A-B>L(v)$ and thus, $f(v)=0$. Now, suppose that $A-B\leq L(v)$. For each color $c\in C$, define $\ell_c=\max(\min(A,B), n^1_c+m_c^v)$. It is evident that $n^1_c\leq \ell_c-m_c^v\leq n^1_c+n^2_c$. So, we can choose $\ell_c-m_c^v$ vertices of color $c$ among children $u_i$'s of $v$ with $g(u_i)=1$ and match it with $u$ such that for any unmatched child $u_j$, we have $f(u_j)=1$. Note that for any two colors $c,c'\in C$, if $A\leq B$, then we have
$\ell_c-\ell_{c'}= A-A= 0$ and if $A\geq B$, then we have 
$\ell_c-\ell_{c'}\leq \max(B,n_c^1+m_c^v)-B\leq A-B\leq L(v)$. Therefore, $v$ is $\tilde{L}(v)$-fair. Finally, for every unmatched child $u_j$, we have $f(u_j)=1$, so we have an $\tilde{L}$-fair left-perfect many-to-one matching for $T_{u_j}$. The union of all these matchings form an $\tilde{L}$-fair left-perfect many-to-one matching for $T_{v}$ and so $f(v)=1$. 

Computing $g(v)$ is very similar to computation of $f(v)$ with only one difference that $v$ should be matched with its parent $p(v)$. Let $p(v)$ be of color $c_0$. We do the same computations as above and just add $n^1_{c_0}$ by one and we have $g(v)=1$ if and only if $A-B\leq L(v)$.  

Using the above recursion, we can find the value of $f$ for the root $u_0$ and it is clear that a solution exists if and only if there is some $F'\subseteq F$ for which $f(u_0)=1$. Also, the runtime of the algorithm is at most $O(2^{\fes}n\log n)$ (because we have to do the above process for each subset of $F$ and in each round, computing the numbers $n^1_c$, $n^2_c$, $A$ and $B$ takes $O(n\log n)$ time). 
\end{proof}

Using a dynamic programming, we can prove that \gfm\ is FPT with respect to $(\tw,\Delta_V)$.
 
\begin{theorem}\label{thm:twDelta}
\gfm\ can be solved in time $O^*(\tw^{\tw} 2^{|\Delta_V|(\tw+3)})$, where $\Delta_V$ is the maximum degree of vertices in $V$ and $\tw$ is the treewidth of the input graph (assuming that nice tree decomposition of width $\tw$ is given).
\end{theorem}

\begin{proof}
Consider an instance $G=(U\cup V,E)$ endowed with the coloring $\col:U\to C$ and a threshold function $L:V\to \mathbb{N}\cup \{0\}$. Also, let $(T,\{B_t\}_{t\in V(T)})$ be a nice tree-decomposition for $G$ of width $\tw$ and $O(\tw |V(G)|)$ nodes.
For each vertex $t\in V(T)$, let $G_t$ be the subgraph of $G$ induced on all vertices in $B_{t'}$ where $t'$ is a descendant of $t$. Define $U^t=U\cap B_t$ and $V^t=V\cap B_t$. Also, for each $u\in U^t$, define a variable $x_u$ as follows.
\[
x_u=\begin{cases}
0 & \text{$u$ is matched with no vertex in $V(G_t)$,}\\
\text{out} & \text{$u$ is matched with a vertex in $V(G_t)\setminus B_t$,}\\
v & \text{$u$ is matched with the vertex $v\in  B_t$.}
\end{cases}
\]
Then, for each $v\in V^t$ and each color $c\in C$, define a variable $y_v^c$ which is equal to the number of vertices in $V(G_t)\setminus B_t$ with color $c$ matched with $v$. Finally, define $z_v^c=|\{u\in B_t: x_u=v, \col(u)=c\}|$, i.e. the number of vertices in $B_t$ with color $c$ matched with $v$. It is clear that $y_v^c+z_v^c\in [0, \Delta_V]$.  Now, define $\mathbf{(x,y)}=(x_u,y_v^c;u\in U^t, v\in V^t, c\in C) $ and the function $f_t$, where $f_t(\mathbf{x,y})$ is equal to one if and only if there exists a many-to-one matching $M$ for $G_t$ such that for all vertices $v\in (V\cap V(G_t))\setminus B_t$, $v$ is $L(v)$-fair, for each vertex $u\in (U\cap V(G_t))\setminus B_t$, $u$ is matched to a vertex and also, for each $u\in U^t$ and $v\in V^t$, $M$ is compatible with $(\mathbf{x,y})$, i.e. $x_u=v$ if $uv\in M$, $x_u=\text{out}$, if $M(u)$ is in $V(G_t)\setminus B_t$ and $x_u=0$ if $M(u)=\emptyset$. Also, $|M_c(v)|=y_v^c+z_v^c$, for all $c\in C$. It is clear that the answer of \gfm\ is yes if and only if $f_r(\emptyset,\emptyset)=1$, where $r$ is the root of $T$ (note that $B_r = \emptyset$). 

First, we enumerate the domain size of $f_t$, i.e. the number of vectors $(\mathbf{x,y})$. Each $x_u$ can take $|V^t|+2$ values, so $\mathbf{x}$ can take $(|V^t|+2)^{|U^t|}$ values. Also, let $v\in V^t$ and let $k$ be the number of color classes in which $v$ has a neighbor. It is clear that $k\leq |\Delta_V|$ and each variable $y_v^c$ can take values in $[0,|N_c(v)|]$.  Hence, the number of possible vectors $(y_v^c;c\in C)$ is at most
\begin{align}
\prod_{c\in C}{\left(|N_c(v)|+1\right)} \leq \left(\dfrac{|\Delta_V|}{k}+1\right)^k \leq 2^{|\Delta_V|}.
\label{eq:AA}
\end{align}

Therefore, the number of vectors $(\mathbf{x,y})$ is at most
\[
(|V^t|+2)^{|U^t|} 2^{|\Delta_V|.|V^t|}\leq O(\tw^{\tw} 2^{|\Delta_V||V^t| }).
\]
First inequality in~\eqref{eq:AA} holds because the product is maximized when all $|N_c(v)|$ are equal and the last inequality in~\eqref{eq:AA} holds because the function $({|\Delta_V|}/{k}+1)^k$ is increasing in terms of $k$. 

 Now, we will recursively find the value of the function $f_t$ on each $(\mathbf{x,y})$ in terms of the value of $f_{t'}$ for the children $t'$ of $t$. We consider the following three possibilities.

\paragraph*{Case 1. $t$ is an introduce node.} Suppose that $t$ has a unique child $t'$, where $B_t=B_{t'}\cup \{w\}$. Firstly, let $w\in U$ and the color of $w$ is $c_0$. If $x_w=\text{out}$, then $f_t(\mathbf{x,y})=0$ (since $w$ has no neighbor in $V(G_t)\setminus B_t$). If $x_w=0$ or $x_w=v$ for some $v\in B_t$, then 
$$f_t(\mathbf{x,y})=f_{t'}(x_u,y_v^c;u\in U^t\setminus\{w\},v\in V^t,c\in C).$$

Secondly, let $w\in V$. If $y_w^c\neq 0$ for some $c\in C$, then $f_t(\mathbf{x,y})=0$ (since $w$ has no neighbor in $V(G_t)\setminus B_t)$.
Also, if $x_u=w$ for some $u$ which is not adjacent to $w$, then evidently $f_t(\mathbf{x,y})=0$. Now, suppose that $y_w^c=0$ for all $c\in C$ and if $x_u=w$, then $u$ is adjacent to $w$. For all $u\in U^t$, define
$$\hat{x}_u=\begin{cases}
0 & x_u=w,\\
x_u & \text{otherwise}.
\end{cases}  
$$
Then, 
$$f_t(\mathbf{x,y})=f_{t'}(\hat{x}_u,y_v^c;u\in U^t,v\in V^t\setminus \{w\},c\in C).$$

Computing $f_t$ for each $(\mathbf{x,y}) $ takes $O(|U^t|)$ time. 

\paragraph*{Case 2. $t$ is a forget node.} Suppose that $t$ has a unique child $t'$, where $B_t=B_{t'}\setminus \{w\}$. 
Firstly, let $w\in U$ and $\col(u)=c_0$. 
Consider all values $x_w\in \{v; v\in N(w)\cap B_t \}\cup \{\text{out}\}$. If $x_w=v_0$, for some $v_0\in N(w)\cap B_t $, then for every $v\in V^t$ and $c\in C$, define 
$$\hat{y}_v^c=\begin{cases}
	y_v^c-1 & v=v_0,c=c_0,\\
	y_v^c & \text{otherwise}.
\end{cases}  
$$
Also, if $x_w=\text{out}$, then define $\hat{y}_v^c=y_v^c$. Now,
$$f_t(\mathbf{x,y})=\max_{x_w\in \{v; v\in N(w)\cap B_t \}\cup \{\text{out}\}} f_{t'}(x_w,x_u,\hat{y}_v^c;u\in U^t,v\in V^t,c\in C).$$

Secondly, let $w\in V$. Let $\tilde{U}$ be the set of all vertices $u\in U^t$ such that $x_u=\text{out}$ and $u$ is adjacent to $w$. For each subset $Y\subseteq \tilde{U}$ and each $u\in U^t$, define
\[
\hat{x}_u(Y)=\begin{cases}
w & u\in Y,\\
x_u& \text{otherwise}.	
\end{cases}
\]
Also, for each $c\in C$, define $n_c=|(N_c(w)\cap V(G_t)) \setminus B_{t'} |$. Now,
if there exist some $Y\subseteq \tilde{U}$ and some integers $y_w^c \in [0,n_c]$, $c\in C$, such that 
\[
\max_{c} (y_w^c+|Y\cap U_c|)- \min_{c} (y_w^c+|Y\cap U_c|) \leq L(w),
\]
and 
\[
f_{t'}(\hat{x}_u(Y), y_v^c,y_w^c;u\in U^t,v\in V^t,c\in C)=1,
\]
then, $f_t(\mathbf{x,y})=1$. Otherwise, we have $f_t(\mathbf{x,y})=0$.

If $w\in U$, then the number of values of $x_w$ is at most $O(|V^t|)$. If $w\in V$, then the size of $\tilde{U}$ is at most $|\Delta_V|$ and by \eqref{eq:AA}, the number of values for $y_w^c$ is at most $2^{|\Delta_V|}$. Therefore, computing $f_t$ for each $(\mathbf{x,y})$ takes time at most $O(|V^t|+|C||U^t| 2^{2|\Delta_V|})= O^*(2^{2|\Delta_V|})$.

\paragraph*{Case 3. $t$ is a join node.} Suppose that $t$ has two children $t',t''$, where $B_t=B_{t'}=B_{t''}$. Now, for each $u\in U^t$, define $x'_u$ and $x''_u$ such that if $x_u\neq \text{out}$, then $x'_u=x''_u=x_u$, and if $x_u=\text{out}$, then $(x'_u,x''_u)=(0,\text{out})$ or $(\text{out},0)$. Also, for each $v\in V^t$ and $c\in C$, define ${y'}_v^{c}$ and ${y''}_v^{c}$ such that ${y'}_v^{c}+{y''}_v^{c}=y_v^c$. Then, we have
\begin{align*}
f_t(\mathbf{x,y})=&\max_{x'_u,x''_u,{y'}_v^{c},{y''}_v^{c}}\\ &\min(f_{t'}(x'_u,{y'}_v^{c};u\in U^t,v\in V^t,c\in C),f_{t''}(x''_u,{y''}_v^{c};u\in U^t,v\in V^t,c\in C)).
\end{align*}

The number of possible values for $x'_u$ and $x''_u$ is at most $2^{|U^t|}$ and the number of possible values for ${y'}^c_v$ and ${y''}^c_v$ is  at most $2^{|\Delta_V|}$ (with an argument similar to \eqref{eq:AA}). Therefore, computing $f_t$ for each $(\mathbf{x,y})$ takes time at most $O(2^{|U^t|+|\Delta_V|})$.

Finally,  we compute the runtime of the whole algorithm. The number of nodes in the tree decomposition is at most $O(\tw n)$, where $n$ is the number of vertices of $G$. Hence, the runtime of the whole algorithm is at most 
\[ 
O^*(\tw^{\tw} 2^{|\Delta_V||V^t|} \max\{2^{2|\Delta_V|}, 2^{|U^t|+|\Delta_V|}\} )\leq O^*(\tw^{\tw} 2^{|\Delta_V|(\tw+3)}).
\] 
\end{proof}

The following theorem asserts that \gfm\ is FPT with respect to neighborhood diversity and thus it is FPT with respect to vertex cover number (this can be viewed as a generalization of the result of \cite{main} which asserts that fair matching is FPT with respect to $|V|$).
\begin{theorem} \label{thm:nd}
	If the input graph has $n$ vertices with neighborhood diversity $\nd$, then \textsc{Generalized Fair Matching} can be solved in $O^*(\nd^{O(\nd)})$ time.
\end{theorem}
 
 \begin{proof}
 Let $G=(U\cup V,E)$ be the input graph along with the coloring $\col:U\to C$ and a threshold function $L:V\to \mathbb{N}\cup \{0\}$. Without loss of generality, we can assume that $G$ has no isolated vertex. Moreover, we can remove all connected components isomorphic to $K_2$ (because for any connected component $uv$ of $G$ with $u\in U$ and $v\in V$, if $L(v)=0$, then the answer is no and if $L(v)\geq 1$, then we can add the edge $uv$ to the matching and $v$ is trivially $L$-fair).
 
 Now, suppose that $(V_1,\ldots, V_{\nd})$ is a partition of $U\cup V$, such that $G[V_i]$ is either a clique or a stable set and for each $i\neq j$, $V_i$ is either complete or incomplete to $V_j$. Since $G$ has no connected component isomorphic to $K_1$ and $K_2$, for each $i$, either $V_i\subseteq U$ or $V_i\subseteq V$.  So, assume that $V_1,\ldots, V_t\subseteq V$ and $V_{t+1},\ldots, V_{\nd}\subseteq U$.
 
 Now, define an auxiliary bipartite graph $G'=(U\cup V',E')$, where $V'=\{v_1,\ldots, v_t\}$ and $u\in U$ is adjacent to $v_i$ in $G'$ if and only if $u$ is complete to $V_i$ in $G$. Also, define the threshold function $L':V'\to \mathbb{N}\cup \{0\}$, as $L'(v_i)=\sum_{v\in V_i} L(v)$. Now, we claim $G$ has an $L$-fair matching if and only if $G'$ has an $L'$-fair matching. 
 
 To see this, note that if $M$ is an $L$-fair matching for $G$, then define $M'$ to be the matching where $uv_i\in M'$ if and only if there is an edge $uv\in M$, for some $v\in V_i$. For every two color classes $U_c$ and $U_{c'}$ and every vertex $v\in V$, we have $|M_c(v)|-|M_{c'}(v)|\leq L(v)$. Therefore, for each $i\in [t]$, we have 
 \[
|M'_c(v_i)|-|M'_{c'}(v_i)|=\sum_{v\in V_i} (|M_c(v)|-|M_{c'}(v)|)\leq \sum_{v\in V_i} L(v)=L'(v_i).
 \]  
Thus, $M'$ is an $L'$-fair matching for $G'$. Now, suppose that $M'$ is an arbitrary $L'$-fair matching for $G'$. Fix an $i\in [t]$. Suppose that $V_i=\{v^1,\ldots v^s\}$ where $\ell_j=L(v^j)$ for each $j\in [s]$. Also, for each $c\in C$, let $|M'_c(v_i)|= L'(v_i) q_c+r_c$, where $0\leq r_c< L'(v_i)$. Define the sequence of integers $(r^1_c,\ldots, r^s_c)$, where $r^1_c=\ell_1$,\ldots, $r^k_c=\ell_k$, $r^{k+1}_c=r_c-\sum_{j=1}^k{\ell_j}$, $r^{k+2}_c=\cdots=r_c^s=0$ (note that if $r_c\leq \ell_1$, then define $r^1_c=r_c, r^2_c=\cdots,r^s_c=0$). Now, partition $M'_c(v_i)$ into $s$ disjoint subsets $M^1_c,\ldots, M^s_c$, where $|M^j_c|= q_c \ell_j+r^j_c$. Define the matching $M$ for $G$ where $M_c(v^j)= M^j_c$. Now, we prove that $M$ is an $L$-fair matching for $G$. For this, let $c,c'\in C$ be two colors and $v^j\in V_i$. We have to prove that $|M_c(v^j)|-|M_{c'}(v^j)|\leq \ell_j$. Since $M'$ is $L'$-fair, we have $|M'_c(v_i)|-|M'_{c'}(v_i)|=L'(v_i)(q_c-q_{c'})+r_c-r_{c'}\leq L'(v_i)$. We consider the following two cases. If $r_c>r_{c'}$, then, $q_c-q_{c'}\leq 0$. So, $|M_c(v^j)|-|M_{c'}(v^j)|= (q_c-q_{c'})\ell_j+r^j_c-r^j_{c'}\leq r^j_c\leq \ell_j $. Now, if $r_c\leq r_{c'}$, then $q_c-q_{c'}\leq 1$ and $r^j_c\leq r^j_{c'}$ for all $j\in [s]$. Therefore,
$|M_c(v^j)|-|M_{c'}(v^j)|= (q_c-q_{c'})\ell_j+r^j_c-r^j_{c'}\leq \ell_j $. Hence, $M$ is an $L$-fair matching. 

Hence, solving the problem for $G$ is reduced to solving the problem for $G'$ which can be done in time $O^*(\nd^{O(\nd)})$ by Theorem~\ref{thm:FPTk}.
 \end{proof}

Note that in bipartite graphs, neighborhood diversity and modular-width are equal. So, Theorem~\ref{thm:nd} is also valid for the modular-width.

For the last result of the paper, applying a result in parametrized complexity of integer programming, we prove that \gfm\ is FPT with respect to $ (\td,|C|) $. 

\begin{theorem} \label{thm:tdC}
If the tree-depth of the input graph is equal to $\td$ and $|C|$ be the number of colors, then \gfm\ can be solved in time 
$O^*(2^{2^{\td(|C|+1)}})$.
\end{theorem}
\begin{proof}
Let $(G=(U\cup V,E),\col, L)$ be an instance of \gfm. We propose an ILP formulation for the problem as follows. For each edge $uv\in E$, we have a variable $x_{uv}$ which is equal to one if and only if $uv$ is within the solution $M$. Also, for each vertex $v\in V$, we define the variables $x_v$ and $y_v$ to be respectively the maximum and the minimum number of vertices of some color $c\in C$ matched to $v$. So, we have the following ILP formulation:

\begin{align}
&\text{ILP2:}\nonumber\\
&\sum_{v\in V} x_{uv} =1, \quad &\forall\ u\in U, \label{eq:ILP1-a}\\
&y_v-x_v\leq L(v), \quad &\forall\ v\in V,\label{eq:ILP1-b} \\
&x_v\leq \sum_{u\in U_c} x_{uv}\leq y_v, \quad &\forall\ v\in V,\ c\in C, \label{eq:ILP1-c}\\
&x_{uv}\in \{0,1\}, \ x_v, y_v\in \mathbb{N}\cup \{0\} \quad  &\forall\ uv\in E, v\in V. \nonumber
\end{align}

Condition~\eqref{eq:ILP1-a} guarantees that the matching is left-perfect and Conditions~\eqref{eq:ILP1-b} and \eqref{eq:ILP1-c} guarantee fairness of the matching. 
Now, we compute the dual graph $\tilde{G}$ of ILP2 as follows. The vertex set of $\tilde{G}$ corresponds to the constraints, so $V(\tilde{G})= U\cup V\cup (V\times [|C|]) $. There are $|U|+|V|$ constraints in \eqref{eq:ILP1-a} and \eqref{eq:ILP1-b} which have no common variable. So $U\cup V$ is a stable set in $\tilde{G}$. 
For each $v\in V$, there are $|C|+1$ constraints in \eqref{eq:ILP1-b} and \eqref{eq:ILP1-c} which share common variables $x_v$ and $y_v$, so $\{v\}\cup (\{v\}\times [|C|])$ forms a clique of size $|C|+1$ in $\tilde{G}$. For each vertex $v\in V$ and $u\in U_c$, for some color $c$, the constraints corresponding to $u$ in \eqref{eq:ILP1-a} and constraints corresponding to $(v,c)$ in \eqref{eq:ILP1-c} have the common variable $x_{uv}$ if and only if $u$ is adjacent to $v$ in $G_c$.  There is no more edges in $\tilde{G}$. The schematic of the graph $\tilde{G}$ is depicted in Figure~\ref{fig2}. Let $G'$ be the graph obtained from $G$ by blowing up each vertex of $V$ to a clique of size $|C|+1$. It is clear that $\tilde{G}$ is a subgraph of $G'$ and so $\td(\tilde{G})\leq \td(G')\leq (|C|+1)\td$. The last inequality is because in the td-decomposition of $G$ with depth $\td$, we can replace each vertex $v\in V$ with a path of $|C|+1$ vertices to obtain a td-decomposition for $G'$ of depth at most  $(|C|+1)\td$ and the topological height at most $\td$. Finally, if $A$ is the coefficient matrix of ILP2, then it is clear that $\|A\|_\infty=1$. By Theorem~\ref{thm:tdDual}, ILP2 can be solved in $O^*((\|A\|_\infty+1)^{2^{\td_D(A)}})= O^*(2^{2^{\td(|C|+1)}})$. Hence, \gfm\ can be solved in FPT time with respect to $\td+|C|$.   
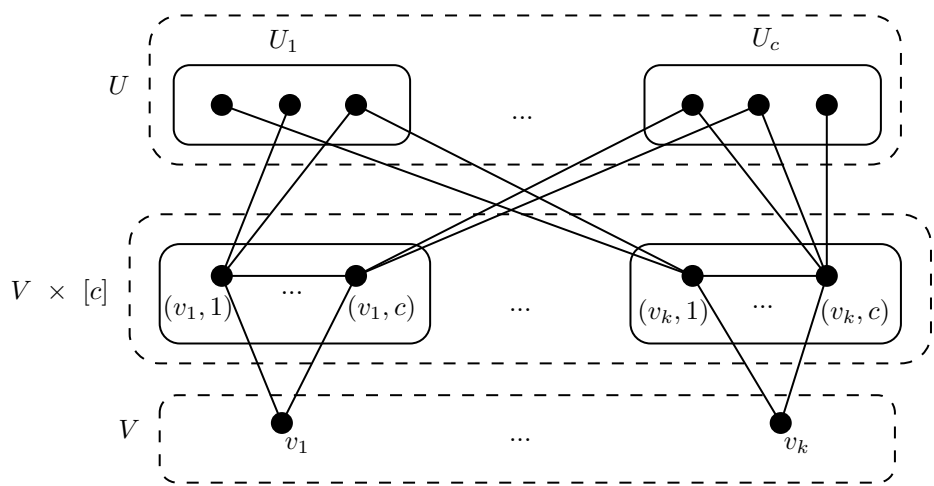
\begin{figure}
	\begin{center}
		\begin{tikzpicture}[x=0.75pt,y=0.75pt,yscale=-1,xscale=1]
			
			\draw   (165,245) .. controls (165,239.48) and (169.48,235) .. (175,235) -- (290,235) .. controls (295.52,235) and (300,239.48) .. (300,245) -- (300,275) .. controls (300,280.52) and (295.52,285) .. (290,285) -- (175,285) .. controls (169.48,285) and (165,280.52) .. (165,275) -- cycle ;
			\draw  [fill={rgb, 255:red, 0; green, 0; blue, 0 }  ,fill opacity=1 ] (191,251) .. controls (191,248.24) and (193.24,246) .. (196,246) .. controls (198.76,246) and (201,248.24) .. (201,251) .. controls (201,253.76) and (198.76,256) .. (196,256) .. controls (193.24,256) and (191,253.76) .. (191,251) -- cycle ;
			\draw  [fill={rgb, 255:red, 0; green, 0; blue, 0 }  ,fill opacity=1 ] (258,251) .. controls (258,248.24) and (260.24,246) .. (263,246) .. controls (265.76,246) and (268,248.24) .. (268,251) .. controls (268,253.76) and (265.76,256) .. (263,256) .. controls (260.24,256) and (258,253.76) .. (258,251) -- cycle ;
			\draw   (400,245) .. controls (400,239.48) and (404.48,235) .. (410,235) -- (525,235) .. controls (530.52,235) and (535,239.48) .. (535,245) -- (535,275) .. controls (535,280.52) and (530.52,285) .. (525,285) -- (410,285) .. controls (404.48,285) and (400,280.52) .. (400,275) -- cycle ;
			\draw  [fill={rgb, 255:red, 0; green, 0; blue, 0 }  ,fill opacity=1 ] (426,251) .. controls (426,248.24) and (428.24,246) .. (431,246) .. controls (433.76,246) and (436,248.24) .. (436,251) .. controls (436,253.76) and (433.76,256) .. (431,256) .. controls (428.24,256) and (426,253.76) .. (426,251) -- cycle ;
			\draw  [fill={rgb, 255:red, 0; green, 0; blue, 0 }  ,fill opacity=1 ] (493,251) .. controls (493,248.24) and (495.24,246) .. (498,246) .. controls (500.76,246) and (503,248.24) .. (503,251) .. controls (503,253.76) and (500.76,256) .. (498,256) .. controls (495.24,256) and (493,253.76) .. (493,251) -- cycle ;
			\draw  [fill={rgb, 255:red, 0; green, 0; blue, 0 }  ,fill opacity=1 ] (221,325) .. controls (221,322.24) and (223.24,320) .. (226,320) .. controls (228.76,320) and (231,322.24) .. (231,325) .. controls (231,327.76) and (228.76,330) .. (226,330) .. controls (223.24,330) and (221,327.76) .. (221,325) -- cycle ;
			\draw  [fill={rgb, 255:red, 0; green, 0; blue, 0 }  ,fill opacity=1 ] (470,325) .. controls (470,322.24) and (472.24,320) .. (475,320) .. controls (477.76,320) and (480,322.24) .. (480,325) .. controls (480,327.76) and (477.76,330) .. (475,330) .. controls (472.24,330) and (470,327.76) .. (470,325) -- cycle ;
			\draw   (172,153) .. controls (172,148.58) and (175.58,145) .. (180,145) -- (282,145) .. controls (286.42,145) and (290,148.58) .. (290,153) -- (290,177) .. controls (290,181.42) and (286.42,185) .. (282,185) -- (180,185) .. controls (175.58,185) and (172,181.42) .. (172,177) -- cycle ;
			\draw  [fill={rgb, 255:red, 0; green, 0; blue, 0 }  ,fill opacity=1 ] (191,165) .. controls (191,162.24) and (193.24,160) .. (196,160) .. controls (198.76,160) and (201,162.24) .. (201,165) .. controls (201,167.76) and (198.76,170) .. (196,170) .. controls (193.24,170) and (191,167.76) .. (191,165) -- cycle ;
			\draw  [fill={rgb, 255:red, 0; green, 0; blue, 0 }  ,fill opacity=1 ] (258,165) .. controls (258,162.24) and (260.24,160) .. (263,160) .. controls (265.76,160) and (268,162.24) .. (268,165) .. controls (268,167.76) and (265.76,170) .. (263,170) .. controls (260.24,170) and (258,167.76) .. (258,165) -- cycle ;
			\draw   (407,153) .. controls (407,148.58) and (410.58,145) .. (415,145) -- (517,145) .. controls (521.42,145) and (525,148.58) .. (525,153) -- (525,177) .. controls (525,181.42) and (521.42,185) .. (517,185) -- (415,185) .. controls (410.58,185) and (407,181.42) .. (407,177) -- cycle ;
			\draw  [fill={rgb, 255:red, 0; green, 0; blue, 0 }  ,fill opacity=1 ] (426,165) .. controls (426,162.24) and (428.24,160) .. (431,160) .. controls (433.76,160) and (436,162.24) .. (436,165) .. controls (436,167.76) and (433.76,170) .. (431,170) .. controls (428.24,170) and (426,167.76) .. (426,165) -- cycle ;
			\draw  [fill={rgb, 255:red, 0; green, 0; blue, 0 }  ,fill opacity=1 ] (493,165) .. controls (493,162.24) and (495.24,160) .. (498,160) .. controls (500.76,160) and (503,162.24) .. (503,165) .. controls (503,167.76) and (500.76,170) .. (498,170) .. controls (495.24,170) and (493,167.76) .. (493,165) -- cycle ;
			\draw  [fill={rgb, 255:red, 0; green, 0; blue, 0 }  ,fill opacity=1 ] (225,165) .. controls (225,162.24) and (227.24,160) .. (230,160) .. controls (232.76,160) and (235,162.24) .. (235,165) .. controls (235,167.76) and (232.76,170) .. (230,170) .. controls (227.24,170) and (225,167.76) .. (225,165) -- cycle ;
			\draw  [fill={rgb, 255:red, 0; green, 0; blue, 0 }  ,fill opacity=1 ] (459,165) .. controls (459,162.24) and (461.24,160) .. (464,160) .. controls (466.76,160) and (469,162.24) .. (469,165) .. controls (469,167.76) and (466.76,170) .. (464,170) .. controls (461.24,170) and (459,167.76) .. (459,165) -- cycle ;
			\draw    (196,251) -- (226,325) ;
			\draw    (263,251) -- (226,325) ;
			\draw    (431,251) -- (475,325) ;
			\draw    (498,251) -- (475,325) ;
			\draw    (196,251) -- (263,251) ;
			\draw    (431,251) -- (498,251) ;
			\draw    (263,165) -- (196,251) ;
			\draw    (464,165) -- (263,251) ;
			\draw    (431,165) -- (263,251) ;
			\draw    (196,165) -- (431,251) ;
			\draw    (230,165) -- (196,251) ;
			\draw    (464,165) -- (498,251) ;
			\draw    (498,165) -- (498,251) ;
			\draw    (263,165) -- (431,251) ;
			\draw    (431,165) -- (498,251) ;
			\draw  [dash pattern={on 4.5pt off 4.5pt}] (165,319.8) .. controls (165,314.94) and (168.94,311) .. (173.8,311) -- (523.2,311) .. controls (528.06,311) and (532,314.94) .. (532,319.8) -- (532,346.2) .. controls (532,351.06) and (528.06,355) .. (523.2,355) -- (173.8,355) .. controls (168.94,355) and (165,351.06) .. (165,346.2) -- cycle ;
			\draw  [dash pattern={on 4.5pt off 4.5pt}] (150,235) .. controls (150,226.72) and (156.72,220) .. (165,220) -- (535,220) .. controls (543.28,220) and (550,226.72) .. (550,235) -- (550,280) .. controls (550,288.28) and (543.28,295) .. (535,295) -- (165,295) .. controls (156.72,295) and (150,288.28) .. (150,280) -- cycle ;
			\draw  [dash pattern={on 4.5pt off 4.5pt}] (160,135) .. controls (160,126.72) and (166.72,120) .. (175,120) -- (520,120) .. controls (528.28,120) and (535,126.72) .. (535,135) -- (535,180) .. controls (535,188.28) and (528.28,195) .. (520,195) -- (175,195) .. controls (166.72,195) and (160,188.28) .. (160,180) -- cycle ;
			
			\draw (218,126) node [anchor=north west][inner sep=0.75pt]    {$U_{1}$};
			\draw (459,125) node [anchor=north west][inner sep=0.75pt]    {$U_{c}$};
			\draw (165,259) node [anchor=north west][inner sep=0.75pt]    {$( v_{1} ,1)$};
			\draw (257,259) node [anchor=north west][inner sep=0.75pt]    {$( v_{1} ,c)$};
			\draw (402,261) node [anchor=north west][inner sep=0.75pt]    {$( v_{k} ,1)$};
			\draw (493,261) node [anchor=north west][inner sep=0.75pt]    {$( v_{k} ,c)$};
			\draw (226,331) node [anchor=north west][inner sep=0.75pt]    {$v_{1}$};
			\draw (475,331) node [anchor=north west][inner sep=0.75pt]    {$v_{k}$};
			\draw (339,169) node [anchor=north west][inner sep=0.75pt]    {$...$};
			\draw (338,266) node [anchor=north west][inner sep=0.75pt]    {$...$};
			\draw (224,257) node [anchor=north west][inner sep=0.75pt]    {$...$};
			\draw (459,264) node [anchor=north west][inner sep=0.75pt]    {$...$};
			\draw (338,331) node [anchor=north west][inner sep=0.75pt]    {$...$};
			\draw (143,322) node [anchor=north west][inner sep=0.75pt]    {$V$};
			\draw (89,251) node [anchor=north west][inner sep=0.75pt]    {$V\ \times \ [ c]$};
			\draw (138,149) node [anchor=north west][inner sep=0.75pt]    {$U$};	
		\end{tikzpicture}
	\end{center}
	\caption{The schematic of the dual graph $\tilde{G}$.}\label{fig2}
\end{figure}
\end{proof}

\section{Concluding Remark}
In this paper, the structural complexity landscape of \gfm\ problem has been investigated. It is proved that the boundary of W[1]-hardness is feedback vertex number and tree-depth. Also, we proved that when we add the number of colors as a parameter, then the boundary is moved from tree-depth to path-width. The complexity class for some parameters are remained unknown, which are listed below.
\begin{itemize}
	\item In Theorems~\ref{thm:fvstdDelta} and \ref{thm:tdC}, we proved that \gfm\ is W[1]-hard with  respect to $(\fvs,\td)$ and FPT with respect to $(\td,|C|)$. So, it is natural to ask if the problem is FPT with respect to $(\fvs,|C|)$?
	\item 
	In Theorem~\ref{thm:twDelta}, we proved that \gfm\ is FPT with respect to ($\tw,\Delta_V$), where both parameters are in the exponent in the running time. This gives rise to the question if the problem is in XP with respect to $\tw$? Also, the W[1]-hardness of the problem with respect to $(\cw,\Delta_V)$ is unknown.
	\item 
	In Theorem~\ref{thm:nd}, we proved that \gfm\ is FPT with respect to $\nd$ and since $\vc \leq O(2^{\nd})$, it is also FPT with respect to $\vc$. Therefore, we can ask if the problem admits a polynomial kernel with respect to $\vc$? Moreover, the parameter vertex integrity is a generalization of vertex cover number, so one may ask if the problem is FPT with respect to $\vi$?
	\item 
	 Finally, all FPT results are also valid for \fm\ problem in which all vertices in $V$ have the same threshold $\ell$. However, in W[1]-hardness results, our reductions construct graphs where the threshold of vertices in $V$ are zero and one. So, it is also valuable if one can find reductions that all vertices have the same threshold (for instance zero). 
\end{itemize}

\end{document}